\documentclass[10pt]{article}
\usepackage{bbm}

\usepackage{amsmath,amssymb,amsthm,color,epsfig,mathrsfs}
\usepackage{cancel}

\newtheorem{theorem}{Theorem}

\newtheorem{claim}[theorem]{Claim}

\newcounter{spslist}

\newcommand{\mat}[5]{ \renewcommand{\arraystretch}{#1}
                    \left[\!\! \begin{array}{cc}
                            #2 & #3 \\
                            #4 & #5 \end{array} \!\!\right] }

\newcounter{geqncount}
    {\refstepcounter{equation}%
     \setcounter{geqncount}{\value{equation}}%
     \setcounter{equation}{0}%
  }%
    {\setcounter{equation}{\value{geqncount}}}

\newcommand{\ZZ}{\mathbb{Z}}
\newcommand{\RR}{\mathbb{R}}
\newcommand{\CC}{\mathbb{C}}

\newcommand{\Vo}{\mathring{\mathcal{V}}}
\newcommand{\VG}{\mathcal{V}(\Gamma)}
\renewcommand{\Vert}{\mathcal{V}}
\newcommand{\EG}{\mathcal{E}(\Gamma)}
\newcommand{\GA}{\Gamma_{\hspace{-2pt}A}}
\newcommand{\GB}{\Gamma_{\hspace{-2pt}B}}
\newcommand{\GC}{\Gamma_{\hspace{-1pt}C}}
\newcommand{\GX}{\Gamma_{\hspace{-2pt}X}}

\begin{document}

\bibliographystyle{plain} 

\begin{center}
{\bf \Large Irreducibility of the Fermi Surface\\\vspace{4pt} for Planar Periodic Graph Operators}
\end{center}

\vspace{0.2ex}

\begin{center}
{\scshape \large Wei Li  \,\,and\,\, Stephen P. Shipman\footnote{Corresponding author; shipman@lsu.edu}} \\
\vspace{1ex}
{\itshape Department of Mathematics, Louisiana State University, Baton Rouge}
\end{center}

\vspace{3ex}
\centerline{\parbox{0.9\textwidth}{
{\bf Abstract.}\
We prove that the Fermi surface of a connected doubly periodic self-adjoint discrete graph operator is irreducible at all but finitely many energies provided that the graph (1) can be drawn in the plane without crossing edges (2) has positive coupling coefficients (3) has two vertices per period.  If ``positive" is relaxed to ``complex", the only cases of reducible Fermi surface occur for the graph of the tetrakis square tiling, and these can be explicitly parameterized when the coupling coefficients are real.  The irreducibility result applies to weighted graph Laplacians with positive weights.
}}

\vspace{3ex}
\noindent
\begin{mbox}
\end{mbox}

\vspace{3ex}

\noindent
\begin{mbox}
\end{mbox}
\vspace{3ex}

\hrule
\vspace{1.1ex}

\section{Introduction}\label{sec:introduction} 

The Fermi surface (or Fermi curve) of a doubly periodic operator at an energy $E$ is the analytic set of complex wavevectors $(k_1,k_2)$ admissible by the operator at that energy.  Whether or not it is irreducible is important for the spectral theory of the operator because reducibility is required for the construction of embedded eigenvalues induced by a local defect \cite{KuchmentVainberg2000,KuchmentVainberg2006,Shipman2014} (except, as for graph operators, when an eigenfunction has compact support).
(Ir)reducibility of the Fermi surface has been established in special situations.  
It is irreducible for the discrete Laplacian plus a periodic potential in all dimensions~\cite{Liu2020} (see previous proofs for two dimensions~\cite{Battig1988},\cite[Ch.~4]{GiesekerKnorrerTrubowitz1993} and three dimensions~\cite[Theorem~2]{Battig1992}) and for the continuous Laplacian plus a potential that is separable in a specific way in three dimensions~\cite[Sec.~2]{BattigKnorrerTrubowitz1991}.
Reducibility of the Fermi surface is attained for multilayer graph operators constructed by appropriately coupling discrete graph operators~\cite{Shipman2014} or metric graph (quantum graph) operators~\cite{Shipman2014,Shipman2019,FisherLiShipman2020a}.

The multilayer graphs in~\cite{Shipman2014,Shipman2019} are inherently non-planar by their very construction---they cannot be rendered in the plane without crossing edges.  This led us to ask whether planarity prohibits reducibility of the Fermi surface.  In this work, we address this problem by direct computation, and we are able to give a complete answer for discrete graph operators with real coefficients and two vertices per period.  It turns out that planarity together with positivity of the coefficients prohibits reducibility.
Our results apply to discrete weighted graph Laplacians (Theorem~\ref{thm:laplacian}), where the positivity of the weights are a discrete version of coercivity of second-order elliptic operators.

Additionally, we find that irreducibility occurs even for complex operator coefficients, including magnetic Laplacians, except for the planar graph whose faces are the tetrakis square tiling (Fig.~\ref{fig:tetrakis}).  For that graph, reducibility occurs for all energies for certain non-positive choices of the coefficients.  All the reducible cases can be explicitly parameterized when the coupling coefficients are real.
Our results are stated in section~\ref{sec:theorems}.  

{\color{black}
Our study of the Fermi surface is part of a more general effort to understand fine properties of the dispersion function $D(k_1,k_2,E)$ of wavevector and energy, particularly using computational methods to attack problems that hitherto evade theoretical methods.
Recently, computational techniques from algebraic geometry~\cite{DoKuchmentSottile2019a} have been used to analyze the genericity of the extrema of the zero set of $D$, also focusing on doubly periodic discrete graph operators with two vertices per period, as in our present work.
}

\section{Periodic graph operators and their Fermi surfaces}\label{sec:setup}

Let $\Gamma$ be a graph with vertex set $\VG$ and edge set $\EG$ drawn in the plane $\RR^2$, such that $\Gamma$ is invariant with respect to shifts along two linearly independent vectors $\xi_1$ and $\xi_2$.  For any $n=(n_1,n_2)\in\ZZ^2$, denote the shift of a vertex $v\in\VG$ or an edge $e\in\EG$ along $n_1\xi_1+n_2\xi_2$ by~$nv$ or~$ne$.  This realizes the shift group $\ZZ^2$ as a free group of isomorphisms of $\Gamma$.  We assume from the beginning that $\VG$ has exactly two orbits; a fundamental domain contains one vertex from each orbit, and the set of these two vertices is denoted by $\Vo=\{v_1,v_2\}$.  Denote the $\ZZ^2$-orbit of $v_i$ by $\Vert_i$, so that $\VG=\Vert_1\mathring\cup\,\Vert_2$.

\subsection{Periodic operators}\label{sec:periodic}

This section sets the background and notation for periodic operators on the space $\ell^2(\VG,\CC)$ of square-integrable functions from the vertex set of $\Gamma$ to the complex field~$\CC$.  Given an operator on $\ell^2(\VG,\CC)$, the edges of the graph $\Gamma$ correspond to nonzero matrix elements of this operator.

Any function $g:\Vo\to\CC$ is identified with the vector $(g_1,g_2)\in\CC^2$ with $g_j = g(v_j)$; and any function $\tilde f\in\ell^2(\VG,\CC)$ is identified with a function $f\in\ell^2(\ZZ^2,\CC^2)$ by $f(n)_j = \tilde f(nv_j)$.  Denote the standard elementary basis vectors in $\CC^2$ by $\{\epsilon_j\}_{j=1}^2$.
A periodic, or shift-invariant, operator $A$ on $\ell^2(\VG,\CC)\cong\ell^2(\ZZ^2,\CC^2)$ is a convolution operator.  Precisely, for each $n\in\ZZ^2$, let $A(n)$ be a $2\!\times\!2$ matrix.  For $f\in\ell^2(\ZZ^2,\CC^2)$, define
\begin{equation}\label{convolution}
  (Af)(n) \;=\; \sum_{m\in\ZZ^2} A(m) f(n-m).
\end{equation}
The operator $A$ is said to be of {\em finite extent} if there is a number $M$ such that $A(n)=0$ for $|n|>M$.
$A$ is self-adjoint if $A(-n)=A(n)^*$ for all $n\in\ZZ^2$.
This class of operators includes the standard and magnetic discrete Laplacians (see~\cite[\S2]{Harper1955} and~\cite[eqn.\,(1.4)]{KorotyaevSaburova2017}).
 
The graph $\Gamma$ is called the {\em graph associated to the operator $A$} provided each entry of $A$, as a matrix indexed by $\VG$, is nonzero if and only if the corresponding pair of vertices is connected by an edge in~$\EG$.  More precisely, for all $i,j\in\{1,2\}$ and all $n\in\ZZ^2$,
\begin{equation}
  \{v_i, nv_j\} \in \EG
  \;\iff\;
  (\epsilon_i,A(n)\epsilon_j) \not=0.
\end{equation}
Double edges do not make sense for discrete graph operators, but loops do.

\subsection{The Fermi surface}\label{sec:fermi}

The $z$-transform (or Floquet transform $\hat f(n,z)$ at $n=0$) of a function $f\in\ell^2(\ZZ^2,\CC^2)$ is the formal Laurent series
\begin{equation}
  \hat f(z) = \sum_{m\in\ZZ^2}f(m)z^{-m},
  \quad z = (z_1,z_2). 
\end{equation}
The notation is $z^m = z_1^{m_1}z_2^{m_2}$, for $m=(m_1,m_2)$.  The Floquet transform of the finite-extent operator $A$ is a Laurent polynomial in $z=(z_1,z_2)$ with $2\!\times\!2$ matrix coefficients,
\begin{equation}
  \hat A(z) = \sum_{m\in\ZZ^2}A(m)z^{-m}.
\end{equation}
Under the Floquet transform, the operator $A$ becomes a matrix multiplication operator,
\begin{equation}
  (Af)\hat{\hspace{2pt}}(z) = \hat A(z)\hat f(z).
\end{equation}
The {\em dispersion function} for $A$ is the determinant
\begin{equation}\label{D}
  D(z_1,z_2,E) \;=\; \det \big(\hat A(z_1,z_2) - E\big).
\end{equation}
It is a Laurent polynomial in $z_1$ and $z_2$ and a polynomial in $E$.
The {\em Floquet surface} $\Phi_E$ for $A$ at energy $E$ is the algebraic set
\begin{equation}
  \Phi_E \;=\; \left\{ (z_1,z_2)\in(\CC^*)^2 : D(z_1,z_2,E)=0 \right\},
\end{equation}
{\color{black}where $\CC^* = \CC\backslash\left\{0\right\}$.}
A pair $(z_1,z_2)$ lies on $\Phi_E$ if and only if $A$ admits a Floquet mode $u$ at $(z_1,z_2,E)$, which is a function $u:\VG\to\CC$ (not in $\ell^2$) such that, for all $v\in\VG$ and $n=(n_1,n_2)\in\ZZ^2$,
\begin{equation}
  Au=Eu
  \quad\text{and}\quad
  u(nv) = z_1^{n_1}z_2^{n_2}u(v),
\end{equation}
that is, $u$ is a simultaneous eigenfunction of $A$ and of the $\ZZ^2$ action.
When considered as a relation between the wavenumbers $(k_1,k_2)\in\CC^2$ (for fixed $E$), where $z_1=e^{ik_1}$ and $z_2=e^{ik_2}$, the Floquet surface is called the {\em Fermi surface}.

The Floquet surface (or the Fermi surface) $\Phi_E$ at energy $E$ is {\em reducible} whenever $D(z_1,z_2,E)$ can be factored nontrivially, that is,
\begin{equation}
  D(z_1,z_2,E) \;=\; D_1(z_1,z_2,E)D_2(z_1,z_2,E),
\end{equation}
in which $D_1$ and $D_2$ are Laurent polynomials in $z_1$ and $z_2$, neither of which is a monomial.
That is to say, reducibility of $\Phi_E$ is the same as nontrivial factorization of $D(z_1,z_2,E)$.

Figure~\ref{fig:labels} shows an example of (one period of) a periodic (nonplanar) graph with two vertices and {\color{black}ten} edges per period, including two loops.  The matrix elements of a self-adjoint periodic operator $A$ with this associated graph are labeled.  The matrix $\hat A(z)-E$ is 
\begin{equation}\label{AzE}
\begin{split}
  \hat A(z)-E\; =\; \mat{1.1}{a_0-E}{b_0}{\bar b_0}{c_0-E}
   &+ z_1\mat{1.1}{a_1}{0}{d_1}{c_1}
     + z_2\mat{1.1}{a_2}{0}{d_2}{0}
     + z_1z_2\mat{1.1}{a_3}{0}{d_3}{0} + \\
   &+ z_1^{-1}\mat{1.1}{\bar a_1}{\bar d_1}{0}{\bar c_1}
     + z_2^{-1}\mat{1.1}{\bar a_2}{\bar d_2}{0}{0}
     + z_1^{-1}z_2^{-1}\mat{1.1}{\bar a_3}{\bar d_3}{0}{0}.
\end{split}
\end{equation}
In section~\ref{sec:reduction}, we show that, by collecting all planar connected periodic graphs into isomorphism classes, it is sufficient to consider only the edges shown in Fig.~\ref{fig:labels}.

\begin{figure}[h]
\centerline{\scalebox{0.25}{\includegraphics{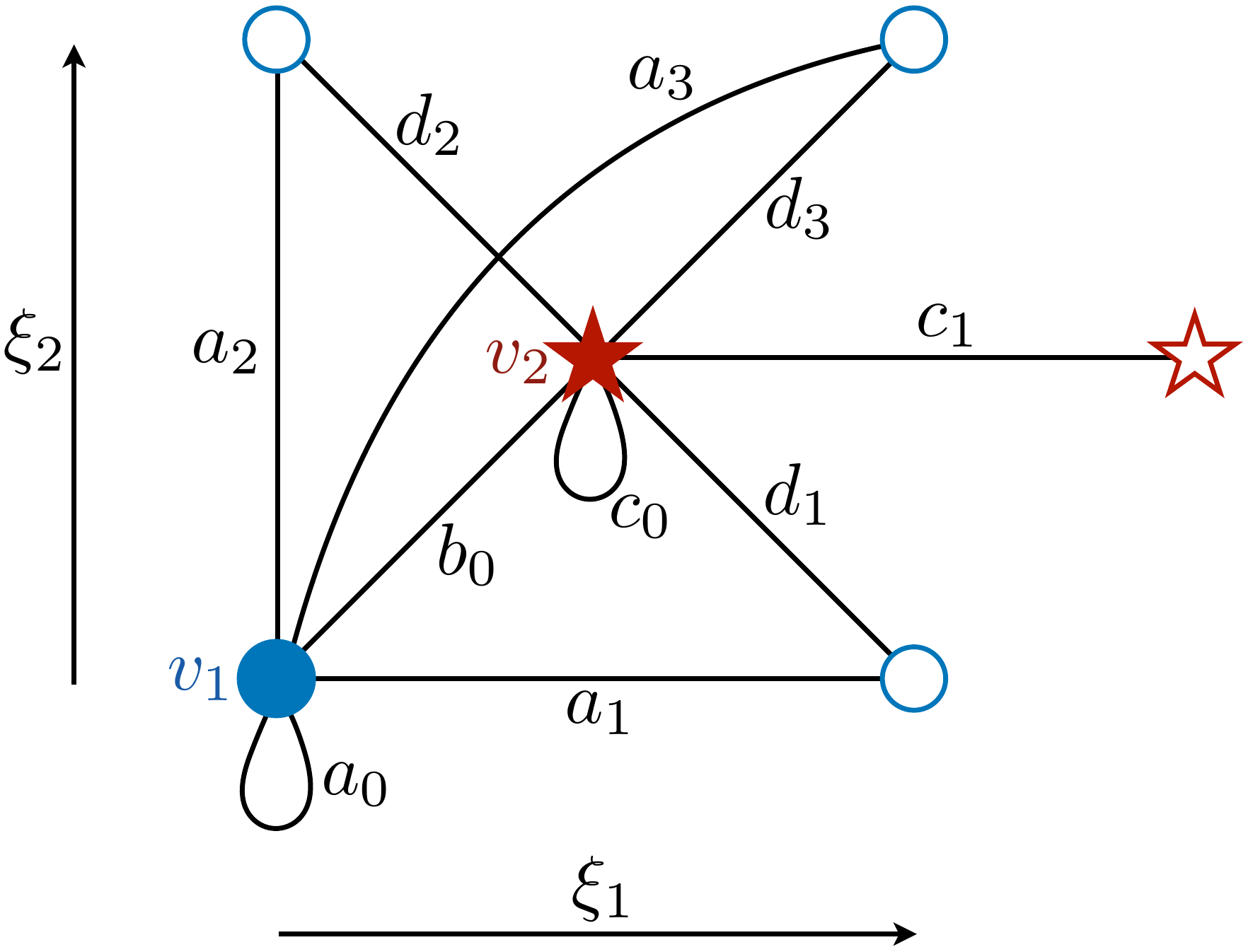}}}
\caption{\small All possible edges in a fundamental domain for the sixteen equivalence classes discussed in Section~\ref{sec:reduction}.}
\label{fig:labels}
\end{figure}

\subsection{Theorems}\label{sec:theorems}


The main result of this work is Theorem~\ref{thm:irreducible}, and Theorem~\ref{thm:laplacian} is an application of it to weighted graph Laplacians with positive coefficients.  Subsequent sections develop their proofs.

\begin{theorem}[Irreducible Fermi surface for planar graphs]\label{thm:irreducible}
Let $A$ be a doubly periodic self-adjoint discrete graph operator (with complex graph coefficients) whose associated graph $\Gamma$ is connected and planar with two vertices per fundamental domain.  The Fermi surface for $A$ is either (i) irreducible for all but a finite number of energies~$E$ or (ii) reducible for all energies~$E$.

(a) If the matrix elements of $A$ corresponding to the edges of $\Gamma$ are real and those corresponding to non-loops are positive, then (i) holds.

(b)
If (ii) holds, then $\Gamma$ minus its loops is the graph of the tetrakis square tiling.
If a factorization of type $D(z_1,z_2)=D_1(z_1)D_2(z_2)$ is admitted, it must be of the form
\begin{equation}\label{d1d2}
  D(z_1,z_2) \;=\; C\big(\alpha_1z_1 + \bar\alpha_1z_1^{-1} + (\beta_1+\gamma_1E)\big)\big(\alpha_2z_2 + \bar\alpha_2z_2^{-1} + (\beta_2+\gamma_2E)\big),
\end{equation}
in which $C$ is a constant, $\alpha_1$ and $\alpha_2$ are complex numbers, and $\beta_i$ and $\gamma_i$ are real numbers.  
If the graph coefficients are real, then (\ref{d1d2}) is the only possible factorization and $\alpha_i$ are real.  The set of all operators with real coefficients and reducible Fermi surface is parameterized explicitly by four real variables.
\end{theorem}

The parameterization of the set of graph operators with real coefficients associated to the tetrakis tiling will be given in section~\ref{sec:algorithm} (equations~(\ref{parameterization2})).

Part (a) of this theorem applies to weighted discrete graph Laplace operators.  Such an operator $A$ acts on a function $\tilde f\in\ell^2(\VG,\CC)$ by
\begin{equation}\label{laplacian}
  (A\tilde f)(v) \;=\; \sum_{e=(w,v)\in\EG} \hspace{-2pt} a_{e} \big( \tilde f(w) - \tilde f(v) \big).
\end{equation}
%
By identifying $\tilde f$ with $f\in\ell^2(\ZZ^2,\CC^2)$ as described in section~\ref{sec:periodic}, $A$ can be written in the convolution form~(\ref{convolution}).  The matrices $A(m) = \left\{ a^m_{ij} \right\}_{\!i,j\in\{1,2\}}$ and the off-diagonal elements of $A(0)$ can be chosen independently through the choice of the weights $a_e$.  Then the diagonal entries of $A(0)$ are determined by the following relations for $i=1,2$:
\begin{equation}
  \hspace{-4pt}\sum_{m\in\ZZ^2} \hspace{-3pt} a^m_{i1} + a^m_{i2} \;=\; 0\,.
\end{equation}
The graph associated with this operator (as described in section~\ref{sec:periodic}) typically has loops ($a^0_{ii}\not=0$) due to the ``self-interactions" coming from the $-\tilde f(v)$ terms in (\ref{laplacian}).  If all the numbers $a_e$ in~(\ref{laplacian}) are positive, then only the loops contribute non-positive (but real) coefficients to $A$.  Thus part (a) of Theorem~\ref{thm:irreducible} yields a corollary.

\begin{theorem}\label{thm:laplacian}
Let $A$ be a doubly periodic discrete graph Laplace operator with positive weights, whose associated graph $\Gamma$ is connected and planar with two vertices per fundamental domain.  Then the Fermi surface for $A$ is irreducible for all but a finite number of energies~$E$.
\end{theorem}

\begin{figure}[h]
\centerline{\scalebox{0.2}{\includegraphics{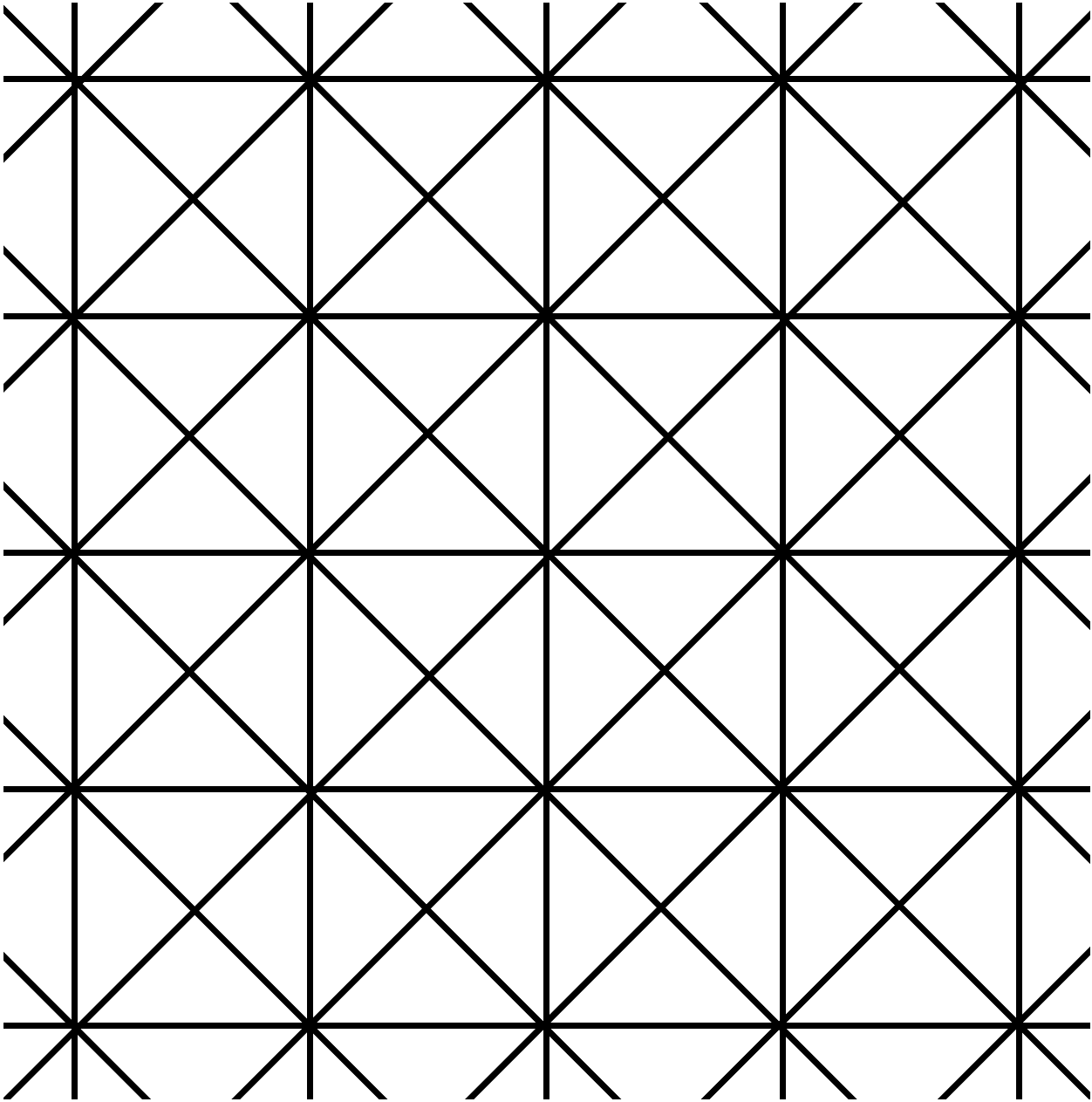}}}
\caption{\small The periodic graph corresponding to the tetrakis square tiling of the plane.}
\label{fig:tetrakis}
\end{figure}


\section{Reduction to sixteen equivalence classes}\label{sec:reduction} 

This section describes the reduction of all graphs associated to the operators we are considering to sixteen cases of equivalent graphs.  Each of these cases is represented by a distinguished graph drawn in the plane whose fundamental domain contains a subset of the edges drawn in Fig.~\ref{fig:labels}.

Let us call a graph {\em admissible} if it is periodic, connected, and planar and has two vertices per fundamental domain and no loops.
We argue in this section that each admissible graph is isomorphic to one of sixteen graphs.  A loop may be added to any of these graphs without losing planarity.
The strategy is to show that each admissible graph has one of three special periodic tilings of the plane as a subgraph, where the boundaries of the tiles contain certain vertices and their shifts.  These special graphs, denoted by $\GA$, $\GB$, and $\GC$, are drawn in Figs.~\ref{fig:GminA} and~\ref{fig:GminBC}.  New edges are then added to the tiles to produce all admissible graphs, which are depicted in Figs.~\ref{fig:GraphsA}, \ref{fig:GraphsB}, and~\ref{fig:GraphsC}.

Recall the set decomposition $\VG=\Vert_1\mathring\cup\Vert_2$ of the vertices of $\Gamma$ into the two disjoint sets of equivalent vertices (the two orbits of the $\ZZ^2$ action).  For $i\in\{1,2\}$, let $\Gamma_i$ denote the subgraph induced by $\Vert_i$, that is, the graph with vertex set $\Vert_i$ and edge set containing those edges in $\EG$ both of whose vertices are in~$\Vert_i$.

There are three disjoint cases to consider.
\begin{itemize}
\addtolength{\itemindent}{3em}
  \item[\bfseries Case A.] $\Gamma$ is bipartite with respect to the decomposition
  $\VG=\Vert_1\mathring\cup\Vert_2$, that is, both $\Gamma_1$ and $\Gamma_2$ have no edges.
  \item[\bfseries Case B.] Either $\Gamma_1$ or $\Gamma_2$ has {\color{black}(or both have)} one edge per fundamental domain, and neither has more than one edge per fundamental domain.
  \item[\bfseries Case C.] Either $\Gamma_1$ or $\Gamma_2$ has at least two edges per fundamental domain.
\end{itemize}
In each of these cases, we determine a minimal subgraph that $\Gamma$ must contain, which will be $\GA$, $\GB$, or $\GC$.  Afterwards, we will find all possible admissible graphs containing that subgraph.

\subsection{Minimal subgraphs}

For each Case $\text{X}\in\{\text{A}, \text{B}, \text{C}\}$, we determine a minimal admissible graph $\GX$.  The graph $\GX$ satisfies the conditions of Case X and has the property that any graph $\Gamma$ satisfying the conditions of Case X must contain $\GX$ up to interchanging $\Vert_1$ and $\Vert_2$ and transformation by a matrix in $GL(2,\ZZ)$.

\medskip
{\bfseries Case A.} This construction is illustrated in Fig.~\ref{fig:GminA}.
$\Gamma$ must contain an edge $e_0$ connecting a vertex $v_1$ in~$\Vert_1$ to a vertex $v_2$ in~$\Vert_2$.  Set $W=\{v_1,v_2\}$.  Since $\Gamma$ is connected, it has an edge $e_1$ connecting a vertex in $W$ with a vertex in a translate $gW$ of $W$ for some nonzero $g\in\ZZ^2$.  
 The translates of the vertices $v_1$ and $v_2$ and the edges $e_0$ and $e_1$ by elements $ng\in\ZZ^2$, with $n\in\ZZ$, form an infinite periodic subchain $C_1$ of $\Gamma$, with vertices alternating between $\Vert_1$ and $\Vert_2$.
{\color{black} The $\ZZ^2$ orbits of $v_1$, $v_2$, $e_0$ and $e_1$ form an array of identical, disjoint copies of $C_1$ that contains all vertices of $\Gamma$.}
Since $\Gamma$ is connected, there is an edge $e_2$ connecting a vertex in $C_1$ with a vertex in a translate of $C_1$ (that is not equal to $C_1$ itself); {\color{black} and since $\Gamma$ is planar, this connection must be between adjacent chains. The $\ZZ^2$ orbits of $v_1$, $v_2$, $e_0$, $e_1$ and $e_2$ is a hexagonal lattice, and contains all vertices of $\Gamma$.}
\begin{figure}[h]
\centerline{
\scalebox{0.3}{\includegraphics{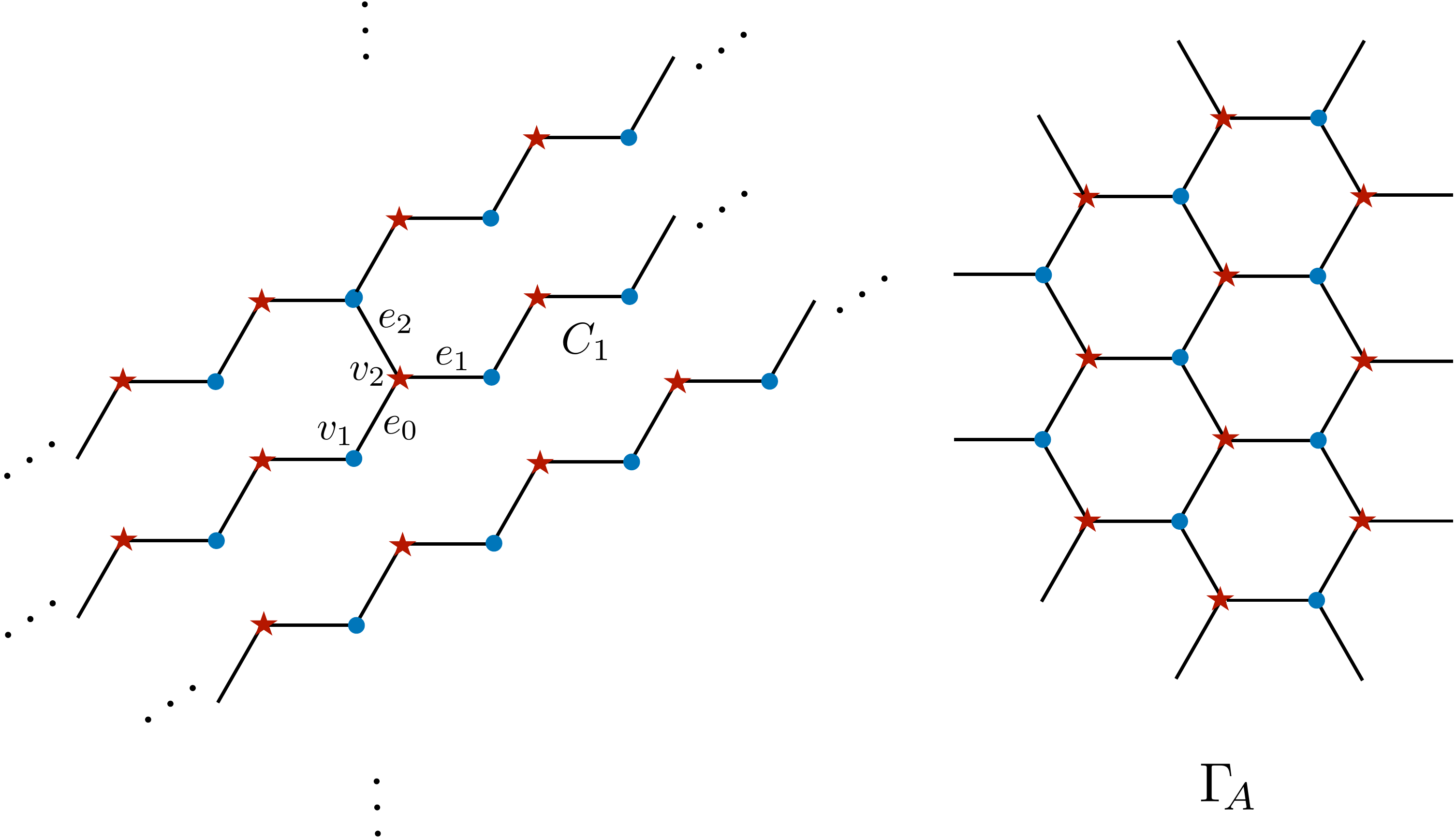}}
}
\caption{\small {\bfseries Case A.}  Every graph $\Gamma$ must have $\GA$ as a subgraph.  The chain $C_1$ is described in the text for the construction of $\GA$.  Vertices in $\Vert_1$ are depicted by blue dots, and vertices in $\Vert_2$ are depicted by red stars.}
\label{fig:GminA}
\end{figure}

\smallskip
{\bfseries Case B.} This construction is illustrated in Fig.~\ref{fig:GminBC} (left).
Let $\Gamma$ have an edge $e_1$ between a vertex $v_1\in\Vert_1$ and $gv_1$ for some nonzero $g\in\ZZ^2$.
The translates of this edge and its vertices by elements $ng\in\ZZ^2$, with $n\in\ZZ$, form a chain $C_1$ connecting a one-dimensional array of vertices in $\Vert_1$.  All the $\ZZ^2$ translates produce an array of disconnected translates of $C_1$, which is just the graph $\Gamma_1$ induced by $\Vert_1$ since Case B disallows more than one edge per fundamental domain of~$\Gamma_1$.
Since $\Gamma$ is connected {\color{black}and planar}, and no edge can join a vertex in $C_1$ with one of its disjoint translates, there is an edge $e_2$ connecting $v_1$ to a vertex $v_2\in\Vert_2$ and another edge~$e_3$ connecting $v_2$ to a vertex $hv_1\in\Vert_1$ in an adjacent translate of $C_1$.  The $\ZZ^2$ orbit of the vertices $v_1$ and~$v_2$ and the edges $e_1$, $e_2$, and $e_3$, form a connected periodic graph $\GB$, as depicted in Fig.~\ref{fig:GminBC}.  As in Case A, since $\Gamma$ is planar, $\GB$ contains all the vertices of $\Gamma$.

\smallskip
{\bfseries Case C.} This construction is illustrated in Fig.~\ref{fig:GminBC} (right).
Let $\Gamma$ have an edge $e_1$ connecting two vertices $v_1$ and $gv_1$ in $\Vert_1$.   
{\color{black} The $\ZZ^2$ orbits of $v_1$,  $gv_1$ and $e_1$ form an array of disjoint chains, which contain all vertices in $\Vert_1$. Since $\Gamma$ is connected and planar, and has an edge from $v_1$ to a vertex in $\Vert_1$ that is different from $v_1$, $gv_1$ and $g^{-1}v_1$, $v_1$ has to be connected to a vertex in $\Vert_1$ on an adjacent chain by an edge $e_2$. The $\ZZ^2$ orbits of $v_1$, $e_1$ and $e_2$ form a square lattice and contain all vertices in $\Vert_1$.}
When $\Gamma$ is drawn in the plane, this grid can be depicted as a square grid with each square containing exactly one vertex from $\Vert_2$ in its interior.
Since $\Gamma$ is connected, a vertex of $\Vert_1$ (and therefore every vertex of $\Vert_1$ by periodicity) must be connected by an edge to a vertex of $\Vert_2$.  Let $e_0$ be an edge connecting $v_1$ to a vertex $v_2\in\Vert_2$, such that $e_0$ does not cross any edge of the square grid.
Define $\GC$ to be this grid together with all the vertices in $\Vert_2$ and the $\ZZ^2$ translates of $e_0$.  This subgraph is depicted in Fig.~\ref{fig:GminBC}.

\begin{figure}[h]
\centerline{
\scalebox{0.34}{\includegraphics{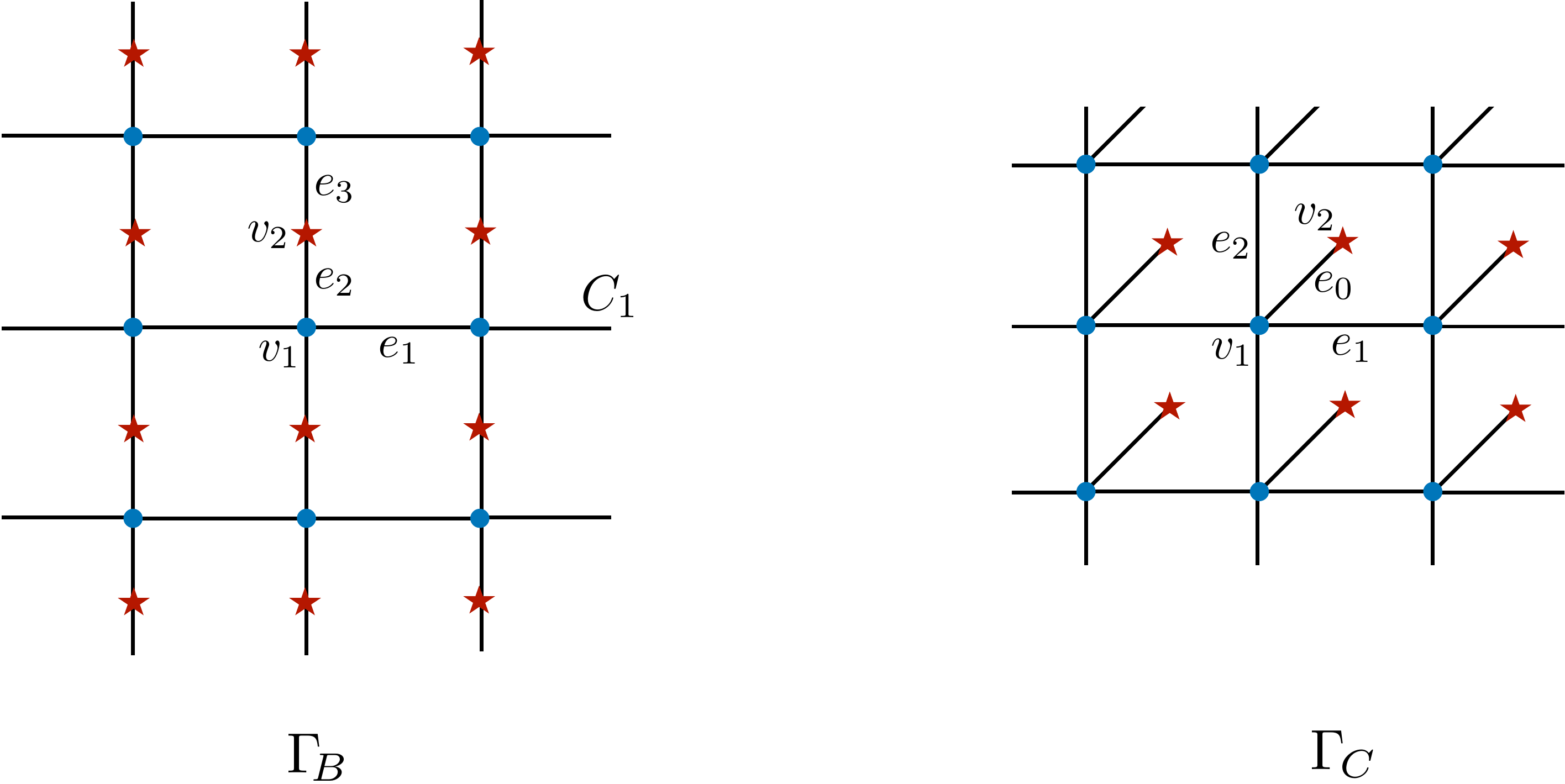}}
}
\caption{\small
{\bfseries Case B.}  Every graph $\Gamma$ must have $\GB$ as a subgraph.  Its construction is described in the text in the construction of $\GB$.
{\bfseries Case C.}  Every graph $\Gamma$ must have $\GC$ as a subgraph. 
Vertices in $\Vert_1$ are depicted by blue dots, and vertices in $\Vert_2$ are depicted by red stars.}
\label{fig:GminBC}
\end{figure}

\subsection{Equivalent graphs for each case}\label{sec:equivalence}

For each of the three cases, we determine up to graph isomorphism all of the planar periodic graphs, with two vertices per period, that contain the given subgraph.

\smallskip
{\bfseries Case A.}
The minimal graph $\GA$ (Fig.~\ref{fig:GminA}) has as a fundamental domain two vertices and three edges---one edge of each of the three orientations on the boundary of a hexagon.  Besides $\GA$ itself, the only way to produce an admissible graph satisfying Case~A and containing $\GA$ is to connect two opposite vertices within the hexagon in Fig.~\ref{fig:GraphsA}.  All three ways of doing this lead to periodic graphs that are isomorphic by rotation.  Thus, the two subcases A1 and A2 in Fig.~\ref{fig:GraphsA} represent the only equivalence classes for Case~A.  Each of these subcases can be identified with the graph in Fig.~\ref{fig:labels}, in which only certain edges are retained.

\begin{figure}[h]
\centerline{
\scalebox{0.4}{\includegraphics{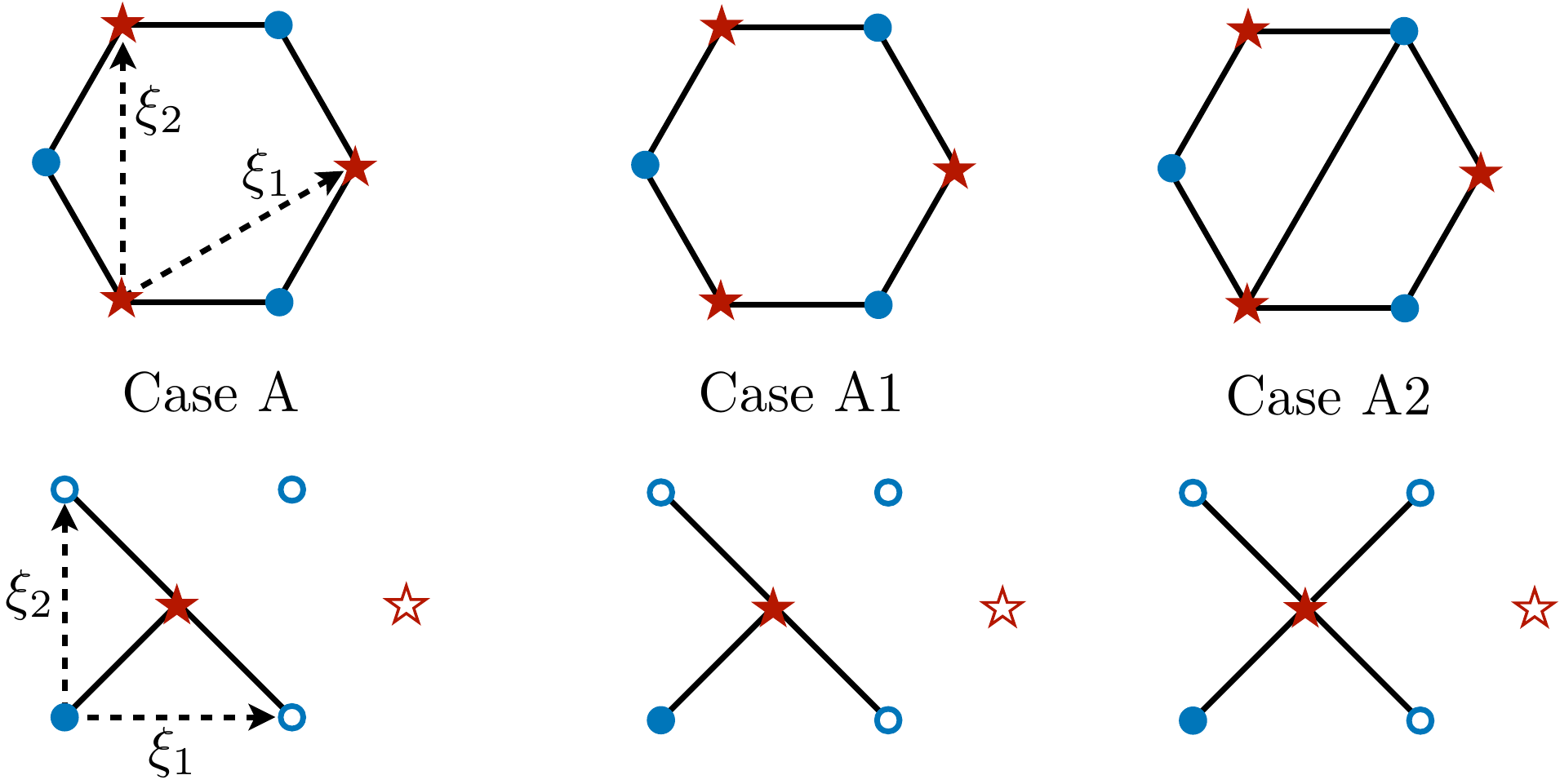}}
}
\caption{\small
The hexagon is one face of the minimal graph $\GA$ shown in Fig.~\ref{fig:GminA} for Case~A.
In each subcase, the edges in one fundamental domain are identified with edges in the standard representation in Fig.~\ref{fig:labels}.
}
\label{fig:GraphsA}
\end{figure}

\smallskip
{\bfseries Case B.} 
The minimal graph $\GB$ (Fig.~\ref{fig:GminBC}) has a fundamental domain consisting of two vertices and three edges on the boundary of a rectangular face, as illustrated in Fig.~\ref{fig:GraphsB}.  All other graphs of Case~B are obtained by joining vertices on the boundary of the face by edges passing through the face provided that no edges cross and no two vertices in $\Vert_1$ (blue circles) are connected.  We argue that there are six equivalence classes of isomorphic graphs, represented by the diagrams in Fig.~\ref{fig:GraphsB}.  First consider the graphs for which the vertices in $\Vert_2$ (stars) are not connected.  There are four ways in which a single edge connects a circle to a star, and all of these are isomorphic by reflections (Case~B2).  There are four ways to have two edges connecting circle to star (Case~B3), and all of them are isomorphic.  The one illustrated in Fig.~\ref{fig:GraphsB} yields only one of the other four upon reflection.  To obtain the other two (both interior edges sloped the same way), one has to shear by the matrix $\mat{0.8}{1}{1}{0}{1}\in GL(2,\ZZ)$ while keeping the two points in any shift of $\Vo$ (circle and star) rigid; this is illustrated in Fig.~\ref{fig:GBiso}.  The other three equivalence classes (Cases B4--B6) are obtained by adding an edge connecting the two stars to each of Cases B1--B3.

\begin{figure}[h]
\centerline{
\scalebox{0.46}{\includegraphics{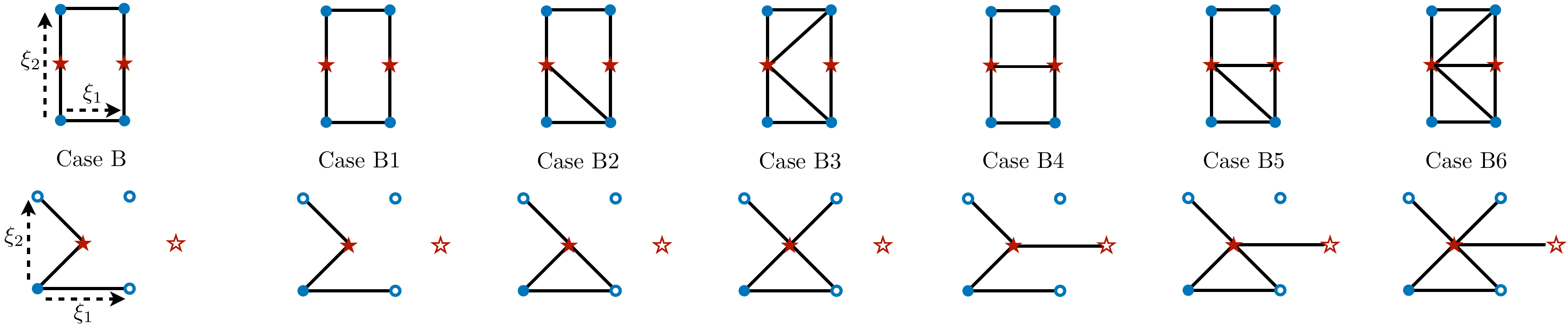}}
}
\caption{\small
The rectangle is one face of the minimal graph $\GB$ shown in Fig.~\ref{fig:GminBC} for Case~B.
In each of the six subcases, the edges in one fundamental domain are identified with edges in the standard representation in Fig.~\ref{fig:labels}.
}
\label{fig:GraphsB}
\end{figure}

\begin{figure}[h]
\centerline{
\scalebox{0.32}{\includegraphics{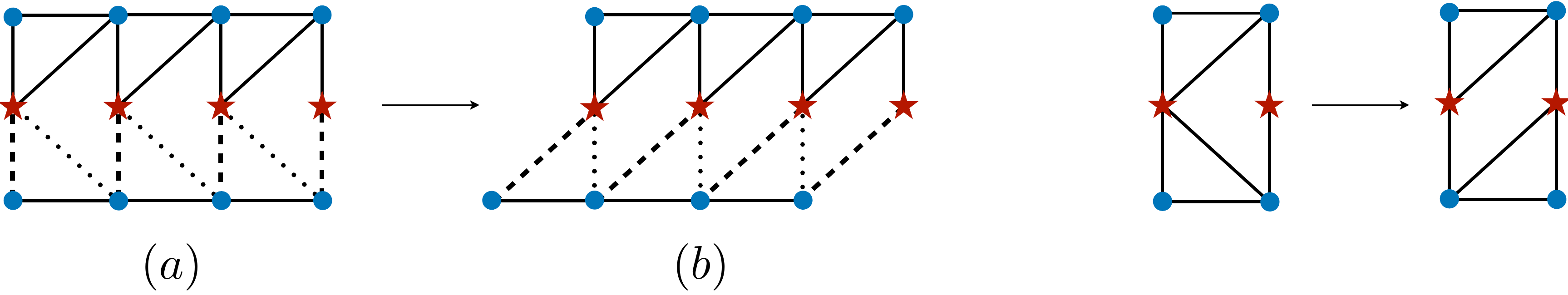}}
}
\caption{\small
In the description of Case~B in section~\ref{sec:equivalence}, it is mentioned how Case~B3 in Fig.~\ref{fig:GraphsB} can be represented by either of the two diagrams on the right through an application of a shear matrix to~$\Gamma$.  This is illustrated by the periodic shift that transforms the drawing (a) to the drawing (b) of the same graph $\Gamma$. 
}
\label{fig:GBiso}
\end{figure}

\smallskip
{\bfseries Case C.}
The minimal graph $\GC$ (Fig.~\ref{fig:GminBC}) has a fundamental domain consisting of {\color{black}two edges} and one vertex of the square, the vertex at the center, and the edge connecting to it, as depicted in Case~C1 of Fig.~\ref{fig:GraphsC}.  We argue that each admissible graph containing $\GC$ as a subgraph is isomorphic to one of the eight cases shown in Fig.~\ref{fig:GraphsC}.
Without adding any diagonal edge, one obtains Cases C1--C4, taking into account equivalence through rotation and reflection, and Case~C8.  The other cases are obtained by adding a diagonal edge to Cases C1--C4.  This results in only three new Cases C5--C7 because in fact Case~C4 with an extra diagonal is isomorphic to Case~C6.  This is seen through applying a shear and then a reflection, similarly to the argument for Case~B.  This is illustrated in Fig.~\ref{fig:GCiso}.

\begin{figure}[h]
\centerline{
\scalebox{0.35}{\includegraphics{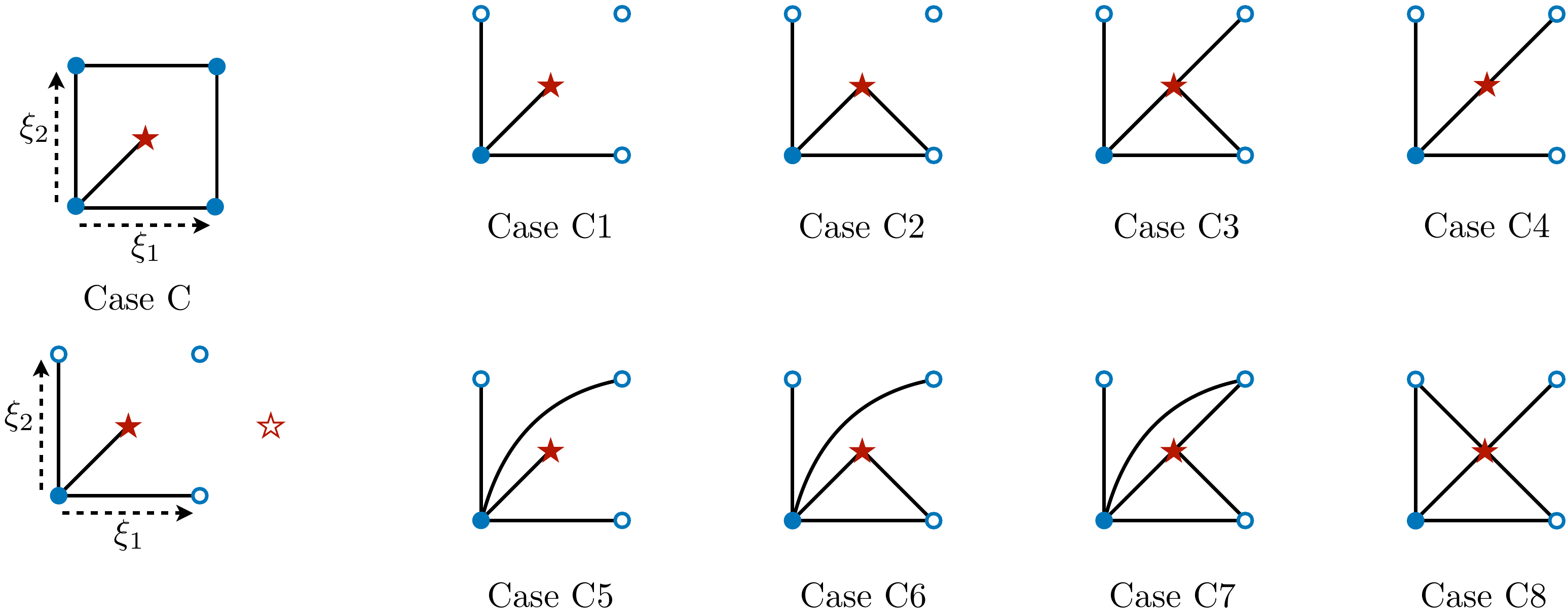}}
}
\caption{\small
The repeating unit of the minimal graph $\GC$ is a face of the subgraph $\Gamma_1$ (grid on circle vertices) plus an extra edge connecting a star vertex.  The edges of each of the eight subcases are drawn according to the standard representation in Fig.~\ref{fig:labels}.
}
\label{fig:GraphsC}
\end{figure}

\begin{figure}[h]
\centerline{
\scalebox{0.35}{\includegraphics{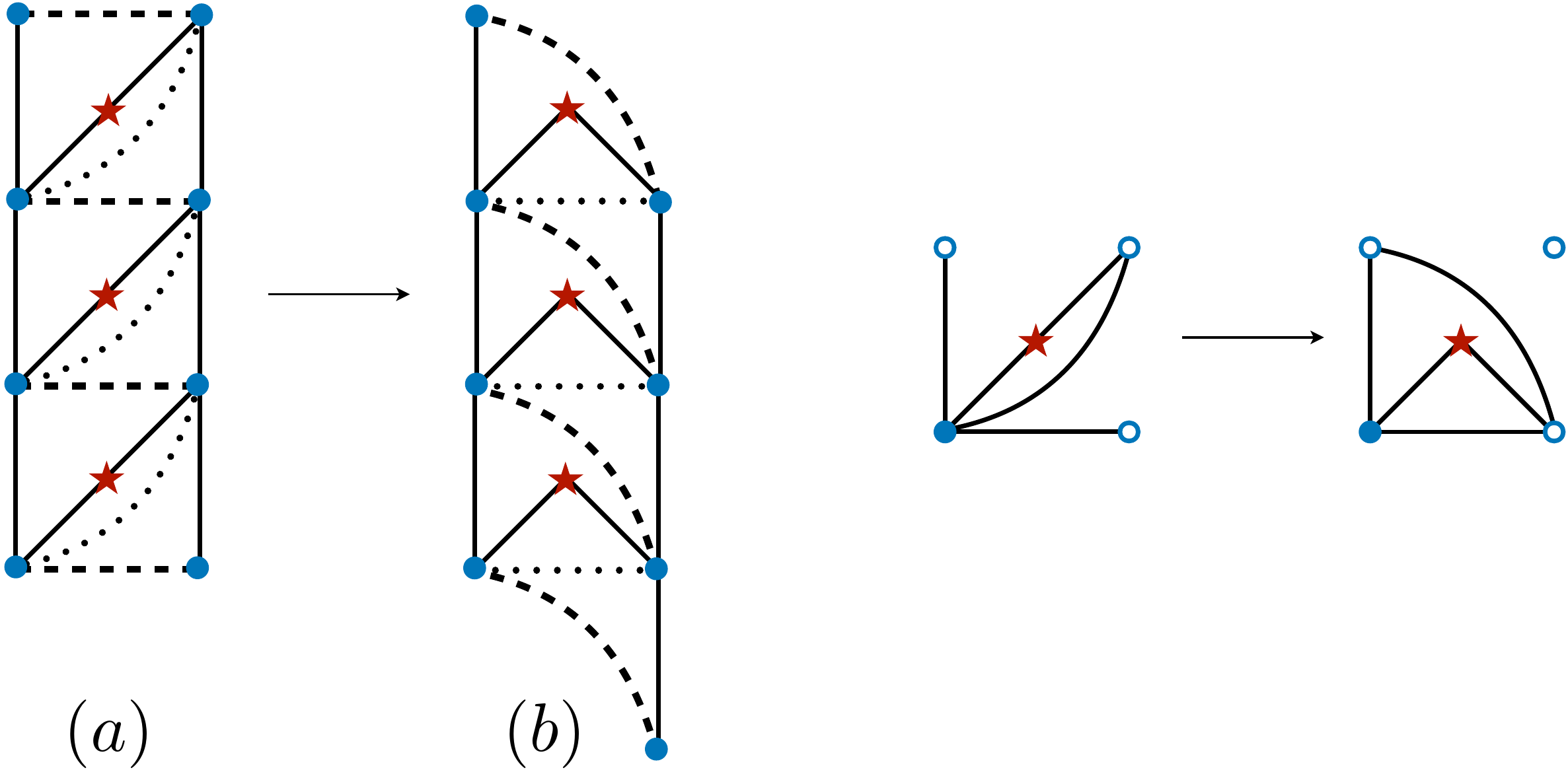}}
}
\caption{\small
In the description of Case~C in section~\ref{sec:equivalence}, it is mentioned how Case~C4 with an extra diagonal is isomorphic to Case~C6.  This is illustrated, up to reflections, by the periodic shift that transforms the drawing (a) to the drawing (b) of the same graph $\Gamma$.
}
\label{fig:GCiso}
\end{figure}

\section{(Ir)reducibility of the Fermi surface} 

Let the graph variables not corresponding to loops be denoted by the multi-variable $S=(a_1,a_2,a_3,b_0,c_1,d_1,d_2,d_3)$, and let $\bar S$ denote the complex conjugates as independent variables.
A graph operator of any of the sixteen types is determined by a choice of values of the variables $a_0,c_0,S$.

We argue that, without sacrificing the generality of our analysis, we may take $c_0=0$.  Indeed, with the alternatives (i) and (ii) in Theorem~\ref{thm:irreducible} in mind, let $S$ be fixed at some arbitrary values.  For any $\alpha\in\CC$, factorability of $D(z_1,z_2,E)$ for all energies $E$ with $a_0=\alpha$ and $c_0=0$ is equivalent to factorability for all energies $E$ for any choice of $a_0$ and $c_0$ with $a_0-c_0=\alpha$; this is evident from the definition of $D$ in (\ref{D}) and (\ref{AzE}).  Similarly, non-factorability of $D(z_1,z_2,E)$ at all but finitely many energies $E$ with $a_0=\alpha$ and $c_0=0$ is equivalent to non-factorability at all but finitely many energies $E$ for any choice of $a_0$ and $c_0$ with $a_0-c_0=\alpha$.  Therefore, we will henceforth take
\begin{equation}
  c_0 \;=\; 0.
\end{equation}

\subsection{Alternatives (i) and (ii) in Theorem~\ref{thm:irreducible}}\label{sec:alternatives}

In this subsection, we justify that either (i) or (ii) in Theorem~\ref{thm:irreducible} holds.

The dispersion function $D(z_1,z_2,E)$ is a Laurent polynomial in $(z_1,z_2)$ with coefficients that are polynomials in
$(E,a_0,S,\bar S)$.  There is a nontrivial factorization $D(z_1,z_2,E)=D_1(z_1,z_2,E)D_{\color{black}2}(z_1,z_2,E)$ at given values of $(E,a_0,S)$ if and only if a certain finite set $P_1$ of polynomials  {\color{black}($P_1$ is $GB_C(I)$ as defined precisely in Step~7 of Section~\ref{sec:algorithm})} in the variables $(E,a_0,S,\bar S)$ vanishes at those values, with the values $\bar S$ being numerically conjugate to the values $S$.

Let $P_2$ be the set of all coefficients of the polynomials in $P_1$, viewed as polynomials in $E$  {\color{black}($P_2$ is $\text{Coeff}_E$ as defined precisely in Step~8 of in Section~\ref{sec:algorithm})}. The elements of $P_2$ are polynomials in $(a_0,S,\bar S)$.  First, let $(a_0,S,\bar S)$ take on values, with $S$ and $\bar S$ being numerical conjugates, at which all elements of $P_2$ vanish.  Then all elements of $P_1$ vanish at all energies $E$, which implies that $D(z_1,z_2,E)$ admits a nontrivial factorization.  Alternatively, let $(a_0,S,\bar S)$ take on values at which some element of $P_2$ does not vanish.  Thus some element of $P_1$ is a nonzero polynomial in $E$.  For any value of $E$ that is not a root of this polynomial, there is then an element of $P_1$ that does not vanish at the values $(E,a_0,S)$.  This implies that $D(z_1,z_2,E)$ is not factorable at such $(E,a_0,S)$.  Thus, in the second alternative, $D(z_1,z_2,E)$ is not factorable at each value of $E$ that is not a root of a certain polynomial.

\subsection{Algorithm}\label{sec:algorithm}
We present an algorithm for determining reducibility of the Fermi surface, which leads to a proof of Theorem~\ref{thm:irreducible}.  We implement it in later sections, with a combination of Mathematica{\textregistered} and by-hand computations.

\smallskip
\noindent
Loop I: do the following steps for each of the sixteen graphs shown in Figs.~\ref{fig:GraphsA}, \ref{fig:GraphsB}, and~\ref{fig:GraphsC}.
\smallskip

1.  Let $\Gamma$ denote one of the sixteen cases shown in Figs.~\ref{fig:GraphsA}, \ref{fig:GraphsB}, and~\ref{fig:GraphsC}. Let $S_{\Gamma}\subset S$ be the set of variables corresponding to edges of $\Gamma$.  That is, the values of the variables in $S_{\Gamma}$ are nonzero but are otherwise undetermined and the values of the variables in $S\backslash S_\Gamma$ are zero. 
Set ${\color{black}K_{\Gamma}} = |S_{\Gamma}|$ and let $\bar S_{\Gamma}$ denote the set of the complex conjugates of the variables in $S_{\Gamma}$, viewed as independent from $S_\Gamma$.
Compute $D(z_1,z_2,E):= \det (\hat A (z) - E)$, where $\hat A(z)$ is defined in~\eqref{AzE} with the variables in $S\backslash S_{\Gamma}$ and their complex conjugates set equal to zero.
$D(z_1,z_2,E)$ is a Laurent polynomial in $z_1$ and $z_2$
with coefficients that are polynomials in the variables $(E,a_0,S_{\Gamma},\bar S_{\Gamma})$. 
It possesses the symmetry
\begin{equation}\label{symmetry}
  D(z_1,z_2) = \bar D(z_1^{-1},z_2^{-1}),
\end{equation}
in which $\bar D$ indicates conjugation of the coefficients $(S_{\Gamma},\bar S_{\Gamma})$ and not conjugation of the variables $(z_1,z_2)$.

2. Compute $P_c: =  z_1^{m_0} z_2^{n_0} D$, where $m_0$ and $n_0$ are the highest powers of $z_1$ and $z_2$ in the monomials of~$D$, which, by (\ref{symmetry}) are also minus the lowest powers of $z_1$ and $z_2$ in the monomials of~$D$.  In all cases we compute below, it turns out that, for all choices of $a_0\in\RR$ and $S_\Gamma\in(\CC^*)^{\color{black}K_{\Gamma}}$, $P_c$ is a polynomial in $z_1$ and $z_2$ with neither $z_1$ nor $z_2$ as a common factor for all but finitely many values of~$E$.
Let $M$ denote the degree of $P_c$. Factoring $D$ nontrivially (neither factor is a monomial) into Laurent polynomials in $z_1$ and $z_2$ is equivalent to factoring $P_c$ into nonconstant polynomials in $z_1$ and $z_2$.

3. The polynomial $P_c$ is of the form $P_c = \sum_{(m,n)\in \mathcal{I}_c} c_{m,n} z_1^mz_2^n$, where $\mathcal{I}_c = \left\{(m,n)\in U: {\color{black}m+n}\leq M\right\}$ and $U=\mathbb Z_{\geq0}\times\mathbb Z_{\geq0}$. 
Compute $c_{m,n}$ as a polynomial in the variables $(E,a_0,S_{\Gamma},\bar S_\Gamma)$. 
Identify the set $C_Z$ of ``always zero" coefficients $c_{m,n}$ and the set $C_N$ of ``almost never zero" coefficients $c_{m,n}$.
More precisely, $(m,n)\in C_Z$ if and only if $c_{m,n}=0$ for all real $E$ and $a_0$ and $S_{\Gamma} \in (\mathbb C^*)^{\color{black}K_{\Gamma}}$, that is, $c_{m,n}$ is identically zero as a polynomial in these variables;
and $(m,n)\in C_N$ if and only if  for any real $a_0$ and $S_{\Gamma} \in (\mathbb C^*)^{\color{black}K_{\Gamma}}$, $c_{m,n}=0$ for at most finitely many real $E$, that is, $c_{m,n}$ is a nonzero polynomial in $E$ for each choice of real $a_0$ and $S_{\Gamma} \in (\mathbb C^*)^{\color{black}K_{\Gamma}}$.

4. The polynomial $P_c$ factors if and only if there exist polynomials $P_a$ and $P_b$ of the form $P_a = \sum_{(m,n)\in \mathcal{I}_a} a_{m,n} z_1^mz_2^n$ and $P_b = \sum_{(m,n)\in \mathcal{I}_b} b_{m,n} z_1^mz_2^n$ such that
\begin{equation}\label{eq:Icomp}
c_{m,n} \;=\; \sum_{(i,j)+(k,l)=(m,n), (i,j)\in \mathcal{I}_a, (k,l)\in \mathcal{I}_b  } a_{i,j}b_{k,l}  \qquad \forall \; (m,n) \in U,
\end{equation} 
in which $\mathcal{I}_a=\mathcal{I}_b = \left\{(m,n)\in U:m+n\leq M\!-\!1 \right\}$. 
To decrease the computational cost, the index sets $\mathcal{I}_a$ and $\mathcal{I}_b$ can be reduced by exploiting the relations (\ref{eq:Icomp}).  Consider all $(m,n)\in C_Z$ such that the right-hand side of (\ref{eq:Icomp}) is a monomial in the $a$ and $b$ coefficients; each of these relations gives a product of coefficients equal to zero.
For each of these $(m,n)$, one of the coefficients (say $a_{i,j}$) in the monomial must be set to zero  and its index removed from the corresponding index set (say, $(i,j)$ removed from $\mathcal{I}_a$).  
There are different ways in which this can be done, and each way yields a reduced pair of index sets $(\mathcal{I}_a,\mathcal{I}_b)$ with a corresponding factorization $P_c=P_aP_b$.  Each of these new pairs of index sets leads to new relations obtained by setting the appropriate coefficients in (\ref{eq:Icomp}) to zero.  The new relations may again contain monomials in the right-hand side, and one can further reduce the index set in different ways.  This process can be continued to obtain a tree of index pairs $(\mathcal{I}_a,\mathcal{I}_b)$ with terminal nodes being those for which the relations (\ref{eq:Icomp}) have no monomial right-hand sides for all $(m,n)\in C_Z$.
Some of these index pairs can be ruled out, namely those for which the right-hand side of (\ref{eq:Icomp}) vanishes identically for some $(m,n)\in C_N$.
In practice, we terminate any part of the tree whenever this occurs.

Let $N_{\Gamma}$ be the number of reduced index pairs. (We find that, for each of the sixteen cases there are zero to six reduced index pairs.)
Obviously, for any fixed real $a_0$ and $S_{\Gamma} \in (\mathbb C^*)^{\color{black}K_{\Gamma}}$,  $P_c$ factors at energy $E$ if and only if $P_c$ factors into $P_a$ and $P_b$ with one of the reduced index pairs at energy $E$.

\smallskip
\noindent
Loop II: do Steps 5--8 for each of the $N_{\Gamma}$ reduced index pairs $(\mathcal{I}_a,\mathcal{I}_b)$ of index sets obtained in Step 4.
\smallskip

5.
Construct the set 
\begin{equation}
I :=  \left\{c_{m,n} - \sum_{(i,j)+(k,l)=(m,n), (i,j)\in \mathcal{I}_a, (k,l)\in \mathcal{I}_b } a_{i,j}b_{k,l} ,  \quad \forall \quad (m,n) \in U \right\}.
\end{equation} 
$I$ is a finite set.
Recall that for given $\left\{c_{m,n}:(m,n)\in \mathcal{I}_c \right\}$, $P_c$ factors into $P_a = \sum_{(m,n)\in \mathcal{I}_a} a_{m,n} z_1^mz_2^n$ and $P_b = \sum_{(m,n)\in \mathcal{I}_b} b_{m,n} z_1^mz_2^n$ if and only if there exist $\left\{a_{m,n}:(m,n)\in \mathcal{I}_a \right\}$ and $\left\{b_{m,n}:(m,n)\in \mathcal{I}_b \right\}$ such that every element of $I$ vanishes.

6. Let $\ell=|\mathcal{I}_a|+|\mathcal{I}_b|$, and fix an ordering $O$ of $\left\{a_{i,j}, b_{k,l}, c_{m,n}: (i,j)\in \mathcal{I}_a, (k,l)\in \mathcal{I}_b, (m,n)\in \mathcal{I}_c \right\}$, where the variables $\left\{c_{m,n}:(m,n)\in \mathcal{I}_c \right\}$ are placed at the end. Compute the Groebner basis of the $(\ell\!+\!1)$th elimination ideal of $I$ with the ordering $O$, and denote it by $GB_C(I)$. By standard nonlinear elimination theory~\cite{CoxLittleOShea2007}, the Groebner basis $GB_C(I)$ is a set of polynomials in $\left\{c_{m,n}:(m,n)\in \mathcal{I}_c \right\}$, whose zero set is the projection of the zero set of $I$ onto $\left\{c_{m,n}:(m,n)\in \mathcal{I}_c \right\}$.
That is,  $P_c$ factors into some polynomials $P_a$ and $P_b$ with reduced index sets $\mathcal{I}_a$ and $\mathcal{I}_b$ if and only if  $\left\{c_{m,n}:(m,n)\in \mathcal{I}_c \right\}$ is a root of all elements of $GB_C(I)$. Note that the $\left\{c_{m,n}:(m,n)\in \mathcal{I}_c \right\}$ being a root of $GB_C(I)$ is independent of the choice of the ordering $O$, as long as $\left\{c_{m,n}:(m,n)\in \mathcal{I}_c \right\}$ are placed at the end.

7. Compute the set $GB$ of polynomials in $(E,a_0,S_\Gamma,\bar S_\Gamma)$, which is obtained from $GB_C(I)$ by replacing each $c_{m,n}:(m,n)\in \mathcal{I}_c $ by its expression as a polynomial in $(E,a_0,S_\Gamma,\bar S_\Gamma)$, according to Step 3. 
Compute $\text{Coeff}_E$, which is defined to be the totality of all coefficients of the elements of $GB$ when treated as polynomials in the variable $E$.
It can be seen that $\text{Coeff}_E$ is a set of polynomials in the variables $(a_0,S_\Gamma,\bar S_\Gamma)$. 

8. Determine whether all elements of $\text{Coeff}_E$ vanish for some allowed $a_0$ and $S_{\Gamma}$, where $S_{\Gamma}$ is an element in $(\mathbb R^+)^{\color{black}K_{\Gamma}}$ or $(\mathbb R^*)^{\color{black}K_{\Gamma}}$ or $(\mathbb C^*)^{\color{black}K_{\Gamma}}$, with $S_\Gamma$ and $\bar S_\Gamma$ related through conjugation.  When a choice of values of $(a_0,S_{\Gamma})$ is a root of all elements of $\text{Coeff}_E$, then $P_c$ factors into $P_a$ and $P_b$ of types $\mathcal{I}_a$ and $\mathcal{I}_b$ for all real~$E$. When a choice of values of $(a_0,S_{\Gamma})$ is not a root of all elements of $\text{Coeff}_E$ then $P_c$ factors into $P_a$ and $P_b$ of types $\mathcal{I}_a$ and $\mathcal{I}_b$ for at most finitely many values of $E$.  This is because, for such $(a_0,S_{\Gamma})$, $GB$ contains a nonzero polynomial $q$ in the variable $E$, whose value will be nonzero for all $E$ except for the finite set of roots of $q$.  (Note that if $\text{Coeff}_E$ never identically vanishes for any real $a_0$ and $S_{\Gamma}\in (\mathbb C^*)^{\color{black}K_{\Gamma}}$, then it never identically vanishes for any real $a_0$ and $S_{\Gamma}\in (\mathbb R^*)^{\color{black}K_{\Gamma}}$ or $S_{\Gamma}\in (\mathbb R^+)^{\color{black}K_{\Gamma}}$.)

\smallskip
Here is a summary of the sets of polynomials involved:
\begin{equation*}
\renewcommand{\arraystretch}{1.2}
\left.
  \begin{array}{ll}
    I & \text{contains the polynomials in $({\color{black}c_{m,n}},a_{i,j},b_{k,l})$ whose vanishing is equivalent to a factorization $P_c=P_aP_b$.} \\
  GB_C(I) & \text{is a Groebner basis for $I$ that eliminates $(a_{i,j},b_{k,l})$ and retains $({\color{black}c_{m,n}})$.} \\
  GB & \text{results from substituting ${\color{black}c_{m,n}}$ in $GB_C(I)$ by its expression in the variables $(E,a_0,S_\Gamma,\bar S_\Gamma)$.} \\
  \mathrm{Coeff}_E & \text{contains the coefficients of the elements of $GB$ as polynomials in $E$; they are polynomials in $(a_0,S_\Gamma,\bar S_\Gamma)$.}
  \end{array}
\right.
\end{equation*}

\subsection{Example for the algorithm}\label{sec:example}
We apply the algorithm to Case C5 as an example. 

\smallskip
\noindent
Loop I.
\smallskip

1. Case C5 has edges $S_\Gamma=\left\{ a_1, a_2, a_3, b_0\right\} \in (\mathbb C^*)^4$, and ${\color{black}K_{\Gamma}}=|S_\Gamma| = 4$. Thus
\begin{equation}\label{eq:opC7}
  \hat A(z_1,z_2)- E \;=\;
  \mat{1.4}
  {-E+a_0 + a_1 z_1+ \bar a_1 z_1^{-1}+ {\color{black}a_2 z_2+ \bar a_2 z_2^{-1}} +a_3z_1z_2+\bar a_3z_1^{-1}z_2^{-1}}
  {b_0  }
  {\bar b_0 }
  {-E }.
\end{equation}
The highest powers of $z_1$ and $z_2$ in $D(z_1,z_2,E)$,
the determinant of \eqref{eq:opC7}, are $1$ and $1$.

2. We obtain $P_c = z_1z_2D$,
\begin{equation}
\begin{aligned}
P_c =& - a_3 z_1^2 z_2^2 E- a_1 z_1^2 z_2 E - a_2 z_1 z_2^2 E + z_1 z_2 (-|b_0|^2 {\color{black}- a_0} E + E^2) - \bar a_2 z_1 E - \bar a_1 z_2 E  - \bar a_3 E .  
\end{aligned}
\end{equation}

3. The degree of $P_c$ is $M=4$.  
The coefficients $c_{m,n}$, as polynomials in  $(E,a_0,S_{\Gamma},\bar S_\Gamma)$, are
\begin{equation}
\begin{aligned}\label{cines}
&c_{2,2}  =  - a_3 E , \quad
c_{2,1}  =  - a_1 E, \quad
c_{1,2}  =  - a_2 E,\quad
c_{1,1}  = -b_0\bar b_0  {\color{black}- a_0} E + E^2, \\
&c_{1,0}  =- \bar a_2 E, \quad
c_{0,1}  = - \bar a_1 E, \quad
c_{0,0}  = - \bar a_3 E.
\end{aligned}
\end{equation}
The index sets of zero and non-zero coefficients are
\begin{align}
C_Z &= \left\{ (4,0), (3,1), (1,3), (0,4), (3,0), (0,3), (2,0), (0,2)  \right\} \cup \left\{ (m,n)\in U: m+n>4  \right\},\\
C_N &= U\backslash C_Z.
\end{align}

4. Using the constraints in $C_Z$ and $\tilde C_N= \left\{ (2,2), (0,0)\right\} \subset C_N$, we obtain $N_{\Gamma}=3$ pairs of reduced index sets $\mathcal{I}_a$ and $\mathcal{I}_b$. The corresponding $P_a$ and $P_b$ are of the forms
\begin{align}
\label{eq:C7fact1}
&P_a = a_{1,1}z_1z_2 + a_{0,0},
\quad P_b = b_{1,1}z_1z_2 + b_{1,0}z_1 + b_{0,1} z_2 + b_{0,0};\\
\label{eq:C7fact2}
&P_a = a_{1,1}z_1z_2+ a_{0,1} z_2 + a_{0,0},
\quad P_b = b_{1,1}z_1z_2 + b_{1,0}z_1  + b_{0,0};\\
\label{eq:C7fact3}
&P_a = a_{0,1}z_2 + a_{0,0},
\quad P_b = b_{2,1}z_1^2z_2 + b_{1,1}z_1z_2 + b_{1,0}z_1 + b_{0,0}.
\end{align}

\smallskip
\noindent
For Loop II, we take Pair~\eqref{eq:C7fact1} as an example.
\smallskip

5. The set of coefficients of $P_c - P_aP_b$ in terms of $\left\{a_{i,j}, b_{k,l}, c_{m,n}: (i,j)\in \mathcal{I}_a, (k,l)\in \mathcal{I}_b, (m,n)\in \mathcal{I}_c \right\}$ is
\begin{equation}
\begin{aligned}\label{excab}
I = & \left\{c_{0,0}-b_{0,0} , c_{0,1}-b_{0,1}, c_{1,0}-b_{1,0}, c_{1,1}- a_{1,1}b_{0,0}, c_{1,2} - a_{1,1} b_{0,1}, c_{2,1}-a_{1,1}b_{1,0}, c_{2,2} - a_{1,1} b_{1,1}\right\}.
\end{aligned}
\end{equation}
{\color{black}Here, $a_{0,0}$ is set equal to $1$ without loss of generality.}

6.  Using the ordering $O = \left\{ a_{1,1}, a_{0,0}, b_{1,1}, b_{1,0}, b_{0,1}, b_{0,0}, c_{2,2}, c_{2,1}, c_{1,2}, c_{1,1}, c_{1,0}, c_{0,1}, c_{0,0}  \right\}$, we compute $GB_C(I)$ in Mathematica\textregistered:
\begin{equation}
\begin{aligned}\label{gbc}
GB_C(I) = & \big\{ -c_{1,0} c_{1,2} + c_{2,1} c_{0,1}, 
c_{0,0} c_{1,2}^2 - c_{1,1} c_{1,2} c_{0,1} + c_{2,2} c_{0,1}^2, \\
&\;\;\; c_{0,0} c_{2,1} c_{1,2} - c_{1,0} c_{1,1} c_{1,2} + c_{1,0} c_{2,2} c_{0,1}, 
c_{0,0} c_{2,1}^2 + c_{1,0}^2 c_{2,2} - c_{1,0} c_{2,1} c_{1,1} \big\}.
\end{aligned}
\end{equation}

7. Substitute the relations in \eqref{cines} for $c_{m,n}$ in \eqref{gbc} to obtain $GB$. 
The set of coefficients of $GB$ as polynomials in $E$ is
\begin{equation}\label{eq:8(2)}
\begin{aligned}
\text{Coeff}_E =& \big\{\, a_1\bar a_1+ a_2\bar a_2, -\bar a_1 a_2\, b_0\bar b_0, 
{\color{black} - a_0} \bar a_1 a_2 - \bar a_1^2 a_3 - a_2^2 \bar a_3,
\bar a_1 a_2, -a_2\bar a_2\, b_0\bar b_0, \\
&\;\;\; {\color{black} - a_0} a_2\bar a_2 - \bar a_1 \bar a_2 a_3 + a_1 a_2 \bar a_3,
a_2\bar a_2, a_1 \bar a_2 b_0\bar b_0, {\color{black}a_0} a_1\bar a_2 - \bar a_2^2 a_3 - a_1^2 \bar a_3, -a_1 \bar a_2 \big\}.
 \end{aligned}
\end{equation}
When seeking common solutions of the polynomials in $\text{Coeff}_E$, we keep in mind that the values of $S_\Gamma$ and $\bar S_\Gamma$ are complex conjugates of each other. 

8. The {\color{black}second} term in $\text{Coeff}_E$ is $- \bar a_1 a_2 |b_0|^2$, which can not be zero when $S_\Gamma \in (\mathbb C^*)^4$. 
This means that, for all choices of $S_\Gamma \in (\mathbb C^*)^4$, one of the elements of $GB$ is a nonzero polynomial in $E$.   So if $E$ is not in the finite set of roots of this polynomial, one of the elements of $GB$ does not vanish at $(E,S_\Gamma)$.
We conclude that, for all real $a_0$ and $S_\Gamma \in (\mathbb C^*)^4$, Case C5 factors into Pair~\eqref{eq:C7fact1} for at most finitely many values of $E$.

Repeating Steps 5--8 for the pair~\eqref{eq:C7fact2} reveals that $\text{Coeff}_E$ contains the element $a_1 a_2 a_3 |b_0|^2$. Thus for all $a_0$ and $S_\Gamma \in (\mathbb C^*)^4$, Case C5 factors into the pair~{\color{black}\eqref{eq:C7fact2}} for at most finitely many values of $E$.  Similarly, for the pair~\eqref{eq:C7fact3}, $\text{Coeff}_E$ contains the element $a_1 a_3 |b_0|^2$, and thus for all $a_0$ and $S_\Gamma \in (\mathbb C^*)^4$, Case C5 factors into the pair \eqref{eq:C7fact3} for at most finitely many $E$ values.

The conclusion in this example is that, for all $a_0$ and $S_\Gamma \in (\mathbb C^*)^4$, Case C5 factors for at most finitely many values of $E$.

\subsection{Proof of Theorem~\ref{thm:irreducible}(a): Irreducibility} \label{sec:proof1a}
Now we present the proof of Theorem~\ref{thm:irreducible} by showing the key results when the Algorithm is performed for each of the sixteen cases. 

\smallskip\noindent
Case C7.  
$S_{\Gamma} = \left\{a_1,a_2,a_3,b_0,d_1,d_3\right\}$. Thus ${\color{black}K_{\Gamma}}=6$ and
\begin{equation}
\begin{aligned}\label{eq:Pc_1}
P_c \;=\; &
(-b_0d_3-a_3E)z_1^2z_2^2 + (-b_0d_1-a_1E)z_1^2z_2 +(-\bar{d_1}d_3-a_2E)z_1z_2^2\\
&+(-|b_0|^2 - |d_1|^2 - |d_3|^2 {\color{black} -a_0 E} +E^2)z_1z_2 + {\color{black} (-d_1\bar{d}_3-\bar{a}_2E)z_1 +(-\bar{b}_0\bar{d_1}-\bar{a_1}E)z_2
-\bar{b}_0\bar{d}_3 - \bar{a}_3E}.
\end{aligned}
\end{equation}

\smallskip\noindent
Case C6.  
$S_{\Gamma}=\left\{a_1,a_2,a_3,b_0,d_1\right\}$ and $P_c$ is equal to \eqref{eq:Pc_1} with $d_3=0$.

\smallskip\noindent
Case C4.  
$S_{\Gamma}=\left\{a_1,a_2,b_0,d_3\right\}$ and $P_c$ is equal to \eqref{eq:Pc_1} with $a_3 =d_1= 0$.

\smallskip\noindent
Case C3.  
$S_{\Gamma}=\left\{a_1,a_2,b_0,d_1,d_3\right\}$ and $P_c$ is equal to \eqref{eq:Pc_1} with $a_3=0$.

\smallskip\noindent
Case B2.  
We use the isomorphic graph shown in Fig.~\ref{fig:GIso}. Then $S_{\Gamma}=\left\{a_1,b_0,d_1,d_3\right\}$ and $P_c$ is equal to \eqref{eq:Pc_1} with $a_2=a_3=0$.

\smallskip\noindent
Case A1.  
We use the isomorphic graph shown in Fig.~\ref{fig:GIso}. Then $S_{\Gamma}=\left\{b_0,d_1,d_3\right\}$ and $P_c$ is equal to \eqref{eq:Pc_1} with $a_1=a_2=a_3=0$.

\begin{figure}[h]
\centerline{
\scalebox{0.5}{\includegraphics{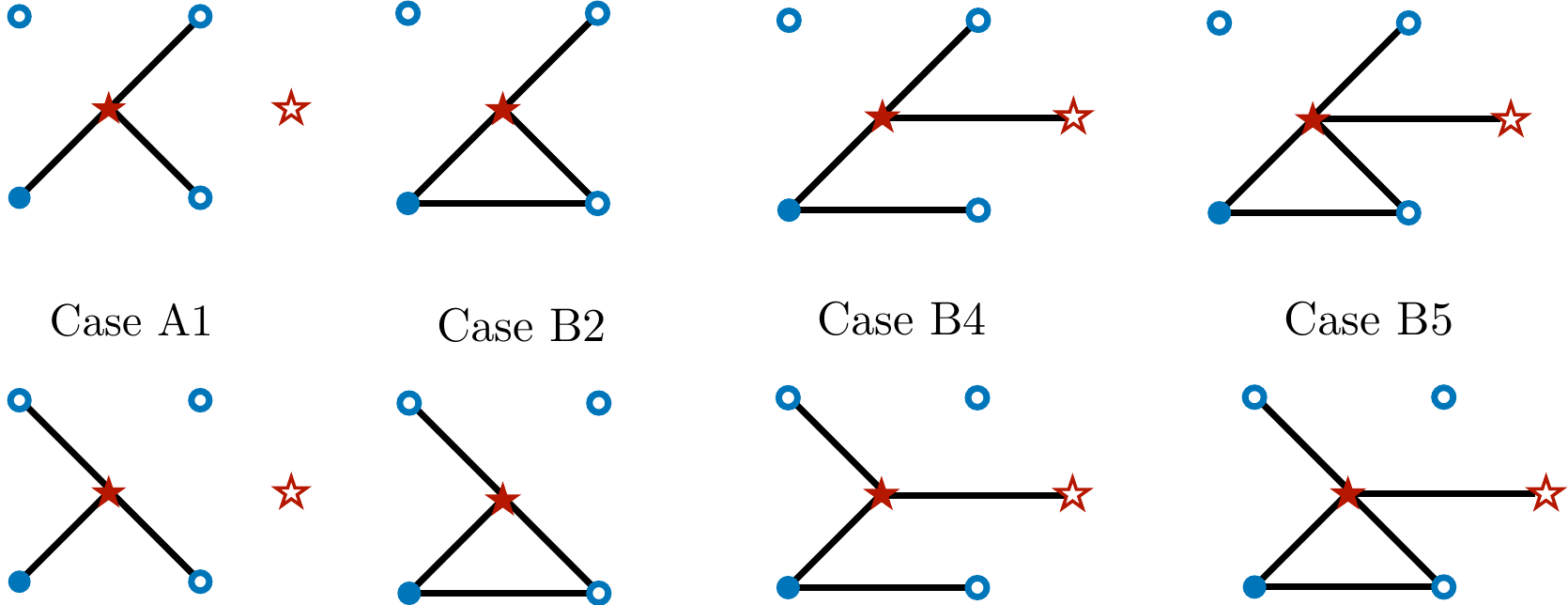}}
}
\caption{\small
These figures of a fundamental domain depict isomorphisms of the corresponding periodic graph, which are used for convenience in the analysis of the dispersion function $D(z_1,z_2;E)$.
}
\label{fig:GIso}
\end{figure}

In each of these cases, the coefficients in (\ref{eq:Pc_1}) do not identically vanish, and thus they all have a common index set~$C_Z$,
\begin{equation}\label{eq:czcn1a}
  C_Z = \left\{ (4,0), (3,1), (1,3), (0,4), (3,0), (0,3), (2,0), (0,2)\right\}  \cup \left\{ (m,n)\in U: m+n>4  \right\}.
\end{equation}
And since $|b_0d_3|^2+|a_3|^2\neq0$ for all these cases, they have a common subset $\tilde C_N$ of $C_N$,
\begin{equation}\label{eq:czcn1b}
\tilde C_N = \left\{ (2,2), (0,0) \right\}.
\end{equation}
The constraints from \eqref{eq:czcn1a} and \eqref{eq:czcn1b} allow three pairs of reduced index sets, whose corresponding $P_a$ and $P_b$ are of the forms \eqref{eq:C7fact1}, \eqref{eq:C7fact2} and \eqref{eq:C7fact3}. Thus, we can compute $\text{Coeff}_E$ for Case C7, and obtain $\text{Coeff}_E$ for the other cases by setting the appropriate variables equal to zero. Table~\ref{tab:nonv1} displays the nonvanishing entries in $\text{Coeff}_E$ for each pair in each of these cases.

\begin{table}[ht]
\centering
\begin{tabular}{|c|c|c|}
\hline
Class &Pair &   non-vanishing element in $\text{Coeff}_E$ for $a_0\in \mathbb R$ and $S_\Gamma \in (\mathbb C^*)^{\color{black}K_{\Gamma}}$\\
\hline
 C7  
  & \eqref{eq:C7fact1}& $\bar b_0 d_1 \bar d_1^3 d_3$\\
& \eqref{eq:C7fact2}&  $\bar a_1 \bar a_2 \bar a_3$ \\
& \eqref{eq:C7fact3}&  $b_0^3 \bar b_0 d_3$ \\
\hline
C6 
& \eqref{eq:C7fact1}&$b_0 \bar b_0 d_1 \bar d_1$\\
& \eqref{eq:C7fact2}&$\bar a_1 \bar a_2 \bar a_3$\\
& \eqref{eq:C7fact3}&$ -\bar a_1 \bar a_3$\\
\hline
 C4 
 & \eqref{eq:C7fact1}&$a_1\bar a_2$\\
& \eqref{eq:C7fact2}& {\color{black}explaned in text}\\
& \eqref{eq:C7fact3}&$-a_1\bar a_1$\\
\hline
C3 
& \eqref{eq:C7fact1}&$-a_1 \bar a_2$\\
& \eqref{eq:C7fact2}& explained in text \\
& \eqref{eq:C7fact3}&${\color{black}\bar d_1^3} \bar b_0 d_1 d_3$\\
\hline
B2 
& \eqref{eq:C7fact1}&$\bar b_0 d_1 \bar d_1^3 d_3$\\
& \eqref{eq:C7fact2}&$\bar a_1 \bar b_0 d_1 \bar d_3^2$\\
& \eqref{eq:C7fact3}&$b_0^3 \bar b_0 d_1 d_3$\\
\hline
A1 
& \eqref{eq:C7fact1}&$\bar b_0 \bar d_1^4 d_3$\\
& \eqref{eq:C7fact2}&$b_0^2 d_1 \bar d_1 d_3^2$\\
& \eqref{eq:C7fact3}&$b_0^3 \bar b_0 d_1 d_3$\\
\hline
\end{tabular}
  \caption{Nonvanishing entries in $\text{Coeff}_E$ for the pairs  \eqref{eq:C7fact1}, \eqref{eq:C7fact2} and \eqref{eq:C7fact3} for 
  Cases C7, C6, C4, C3, B2, and~A1.}
  \label{tab:nonv1}
\end{table}

{\color{black}
For Case C4,
$\text{Coeff}_E$ for pair~\eqref{eq:C7fact2} does not directly contain an element that is nonzero for 
$a_0\in\mathbb R$ and $\left\{a_1,a_2,b_0,d_3\right\}\in (\mathbb C^*)^4$. We show that two particular elements of $\text{Coeff}_E$ cannot be simultaneously zero for  $a_0\in\mathbb R$ and $\left\{a_1,a_2,b_0,d_3\right\}\in (\mathbb C^*)^4$.
These elements are
\begin{equation}\label{eq:coeff7C4}
\begin{aligned}
\left\{ a_1^2 a_2^2 + a_1 a_2 b_0 d_3,
-a_1 a_2 b_0^2 \bar b_0 d_3 + a_1 \bar a_1 b_0^2 d_3^2 + a_2 \bar a_2 b_0^2 d_3^2 - 
 a_1 a_2 b_0 d_3^2\bar d_3 \right\}.
 \end{aligned}
 \end{equation}
 The vanishing of the first element in  \eqref{eq:coeff7C4} requires that $a_1 = -b_0 d_3/a_2$ and $\bar a_1 = -\bar b_0 \bar d_3/\bar a_2$. With these constraints, the second element in \eqref{eq:coeff7C4}  becomes $b_0^2d_3^2 (|a_2|^2 + |b_0|^2)  (|a_2|^2 + |d_3|^2)/|a_2|^2 $, which is nonvanishing for for  $a_0\in\mathbb R$ and $\left\{a_1,a_2,b_0,d_3\right\}\in (\mathbb C^*)^4$.
 }

For Case C3, $\text{Coeff}_E$ for pair~\eqref{eq:C7fact2} does not directly contain an element that is nonzero for 
$a_0\in\mathbb R$ and $\left\{a_1,a_2,b_0,d_1,d_3\right\}\in (\mathbb C^*)^5$. We show that three particular elements of $\text{Coeff}_E$ cannot be simultaneously zero for  $a_0\in\mathbb R$ and $\left\{a_1,a_2,b_0,d_1,d_3\right\}\in (\mathbb C^*)^5$.
These elements are
\begin{equation}\label{eq:coeff7C3}
\begin{aligned}
\big\{ &\bar a_1^2 \bar a_2^2 + \bar a_1 \bar a_2 \bar b_0 \bar d_3,\\
&  2 \bar a_1 \bar a_2^2 \bar b_0 \bar d_1 {\color{black}- a_0} \bar a_1 \bar a_2 \bar b_0 \bar d_3 + 2 \bar a_1^2 \bar a_2 d_1 \bar d_3 + 
 \bar a_2 \bar b_0^2 \bar d_1 \bar d_3 + \bar a_1 \bar b_0 d_1 \bar d_3^2, \\
 &a_2^2 b_0^2 d_1^2 - a_1 a_2 b_0^2 \bar b_0 d_3 {\color{black} - a_0} a_2 b_0^2 d_1 d_3 + 3 a_1 a_2 b_0 d_1 \bar d_1 d_3 + a_1 \bar a_1 b_0^2 d_3^2 + a_2 \bar a_2 b_0^2 d_3^2 {\color{black} - a_0} a_1 b_0 \bar d_1 d_3^2 \\
 &+ b_0^2 d_1 \bar d_1 d_3^2 + a_1^2 \bar d_1^2 d_3^2 - a_1 a_2 b_0 d_3^2 \bar d_3 \big\}.
 \end{aligned}
 \end{equation}
The vanishing of the first two elements in  \eqref{eq:coeff7C3} requires that
$b_0 = -a_1 a_2/d_3$, $a_0 = {\color{black}( a_2^2 d_1 - \bar d_1 d_3^2 )}/(a_2 d_3)$. With these constraints, the third element in  \eqref{eq:coeff7C3} becomes $a_1^2 a_2^2 (|a_1|^2  + |d_3|^2) (|a_2|^2 + |d_3|^2)/(d_3 \bar d_3)$, which is nonvanishing for all  $a_0\in\mathbb R$ and $\left\{a_1,a_2,b_0,d_1,d_3\right\}\in (\mathbb C^*)^5$.

\smallskip
In Case C8,
$S_{\Gamma}=\left\{a_1,a_2,b_0,d_1,d_2,d_3\right\}$. Thus ${\color{black}K_{\Gamma}}=6$ and
\begin{equation}\label{eq:Pc8}
\begin{aligned}
P_c&=(-b_0d_3)z_1^2z_2^2+(-b_0d_1-\bar{d_2}d_3-a_1E)z_1^2z_2 +(-b_0d_2-\bar{d_1}d_3-a_2E)z_1z_2^2\\
&+(-d_1\bar{d_2})z_1^2+(-|b_0|^2 - |d_1|^2 -|d_2|^2 - |d_3|^2 {\color{black} - a_0 E} +E^2)z_1z_2 + (-\bar{d_1}d_2)z_2^2\\
&+(-\bar{b_0}\bar{d_2} -d_1\bar{d_3}-\bar{a_2}E)z_1 +(-\bar{b_0}\bar{d_1} -d_2\bar{d_3} -\bar{a_1}E)z_2
-\bar{b_0}\bar{d_3} 
\end{aligned}
\end{equation}

For case B3, 
$S_{\Gamma}=\left\{ {\color{black}a_1},b_0,d_1,d_2,d_3\right\}$ and $P_c$ is equal to \eqref{eq:Pc8} with ${\color{black}a_2=0}$.

For case A2, 
$S_{\Gamma}=\left\{b_0,d_1,d_2,d_3\right\}$ and $P_c$ is equal to \eqref{eq:Pc8} with $a_1=a_2=0$. 

In each of these three cases, the coefficients in (\ref{eq:Pc8}) do not identically vanish, and thus they all have a common index set~$C_Z$,
\begin{equation}
  C_Z = \left\{ (4,0), (3,1), (1,3), (0,4), (3,0), (0,3)\right\}  \cup \left\{ (m,n)\in U: m+n>4  \right\},
\end{equation}
and since $b_0d_3\neq0$ in these three cases, they have a common subset $\tilde C_N$ of $C_N$,
\begin{equation}\label{eq:czcn8}
\begin{aligned}
  \tilde C_N = \left\{ (2,2), (0,0) \right\}.
\end{aligned}
\end{equation}
The constraints from \eqref{eq:czcn8} allow three pairs of reduced index sets, corresponding to factorizations $P_aP_b$ of the forms
\begin{align}
\label{eq:C8fact1}
&\left(a_{1,1}z_1z_2 +a_{1,0}z_1+ a_{0,1}z_2 + a_{0,0}\right)
\left(b_{1,1}z_1z_2 +b_{1,0}z_1+ b_{0,1}z_2 + b_{0,0} \right),\\
\label{eq:C8fact2}
&\left(a_{0,2}z_2^2 + a_{0,1}z_2 + a_{0,0}\right)
\left(b_{2,0}z_1^2 + b_{1,0}z_1 + b_{0,0}\right),\\
\label{eq:C8fact3}
&\left(a_{0,1}z_2 + a_{0,0}\right)
\left(b_{2,1}z_1^2 z_2 + b_{2,0}z_1^2  + b_{1,1}z_1z_2 + b_{1,0}z_1 + b_{0,1} z_2 + b_{0,0}\right).
\end{align}
Therefore, we can compute $\text{Coeff}_E$ for Case C8 and obtain $\text{Coeff}_E$ for the other two cases by setting the corresponding variables equal to zero. Table \ref{tab:nonv2} displays the nonvanishing entries in $\text{Coeff}_E$ for each pair in each case.  The determinant $D$ in Case C8 may factor for all $E$ for some $a_0\in \mathbb R$ and $S_\Gamma \in (\mathbb C^*)^{\color{black}K_{\Gamma}}$; in fact it also may factor for all $E$ for some $a_0\in \mathbb R$ and $S_\Gamma \in (\mathbb R^*)^{\color{black}K_{\Gamma}}$, but it only factors for at most finitely many $E$ for $a_0\in \mathbb R$ and $S_\Gamma \in (\mathbb R^+)^{\color{black}K_{\Gamma}}$. This is discussed in more detail in Sec.~\ref{sec:irreducible}. The items listed for Case C8 in Table \ref{tab:nonv2} are items that do not vanish for any $a_0\in \mathbb R$ and $S_\Gamma \in (\mathbb R^+)^{\color{black}K_{\Gamma}}$.

\begin{table}[ht]
\centering
\begin{tabular}{|c|c|c|}
\hline
Class &Pair &   nonvanishing element in $\text{Coeff}_E$ for $a_0\in \mathbb R$ and $S_\Gamma \in (\mathbb R^+)^{\color{black}K_{\Gamma}}$\\
\hline
C8& {\color{black}\eqref{eq:C8fact1}}&$-a_1^2 a_2^3-a_1 a_2^2 d_1 d_2-a_1 a_2^2 b_0 d_3-a_2 b_0 d_1 d_2 d_3$\\
& {\color{black}\eqref{eq:C8fact2}}&$a_1 a_2+d_1 d_2$\\
& {\color{black}\eqref{eq:C8fact3}}&$a_1^2 a_2^2+a_1 a_2 d_1 d_2+a_1 a_2 b_0 d_3+b_0 d_1 d_2 d_3$\\
\hline
Class &Pair &   nonvanishing element in $\text{Coeff}_E$ for $a_0\in \mathbb R$ and $S_\Gamma \in (\mathbb C^*)^{\color{black}K_{\Gamma}}$\\
\hline
B3& {\color{black}\eqref{eq:C8fact1}}&$\bar a_1^2 \bar b_0 d_1 \bar d_2 \bar d_3$\\
&{\color{black} \eqref{eq:C8fact2}}&$-\bar b_0 \bar d_3$\\
&{\color{black}\eqref{eq:C8fact3}}&$\bar b_0 \bar d_1 d_2 \bar d_3$\\
\hline
  A2& {\color{black}\eqref{eq:C8fact1}}&$\bar b_0 d_1 \bar d_2 \bar d_3$\\
& {\color{black}\eqref{eq:C8fact2}}&$-\bar b_0 \bar d_3$\\
&{\color{black} \eqref{eq:C8fact3}}&$\bar b_0 \bar d_1 d_2 \bar d_3$\\
\hline
\end{tabular}
  \caption{Nonvanishing entries in $\text{Coeff}_E$ for Pairs  \eqref{eq:C8fact1}, \eqref{eq:C8fact2} and \eqref{eq:C8fact3} for Cases C8, B3, and A2.}
  \label{tab:nonv2}
\end{table}

Case B6:  $S_{\Gamma}=\left\{a_1,b_0,d_1,d_2,d_3, c_1\right\}$. Thus ${\color{black}K}=6$ and
\begin{equation}\label{eq:Pc_11}
\begin{aligned}
P_c\;=\; &a_1c_1z_1^4z_2 - b_0d_3z_1^3z_2^2+( {\color{black} a_0 c_1} -b_0d_1-\bar d_2d_3-(a_1+c_1)E)z_1^3z_2 +(-b_0d_2-\bar d_1d_3)z_1^2z_2^2\\
-&d_1\bar d_2z_1^3  +(-|b_0|^2 +\bar a_1c_1 + a_1\bar c_1 -|d_1|^2 -|d_2|^2-|d_3|^2  {\color{black} -a_0 E} +E^2)z_1^2z_2 - \bar{d_1}d_2 z_1z_2^2\\
 +& (  {\color{black} a_0 \bar c_1} -\bar b_0\bar d_2-d_1\bar d_3)z_1^2 +(-\bar b_0\bar d_1-d_2\bar{d}_3-(\bar{a}_1+\bar{c}_1)E)z_1z_2\\
-&\bar{b}_0\bar{d}_3z_1 + \bar{a}_1\bar{c}_1z_2.
\end{aligned}
\end{equation}
Since {\color{black}$a_1c_1\neq0$ and $b_0d_3\neq0$}, we obtain $C_Z$ and a subset $\tilde C_N$ of $C_N$:
\begin{equation}\label{eq:czcn11}
\begin{aligned}
&C_Z = \left\{ (5,0), (2,3), (1,4), (0,5), (4,0), (1,3), (0,4), (0,3), (0,2), (0,0)\right\}  \cup \left\{ (m,n)\in U: m+n>5  \right\} \\ 
&\tilde C_N = \left\{ (4,1), (3,2), (1,0), (0,1) \right\}.
\end{aligned}
\end{equation}
The constraints from \eqref{eq:czcn11} yield six pairs of reduced index sets, whose corresponding factorizations $P_aP_b$ are of the forms
\begin{align}
\label{eq:11fact1}
&\left(a_{1,1}z_1z_2 + a_{0,0}\right)
\left(b_{3,0}z_1^3 + b_{2,1}z_1^2 z_2 +b_{2,0}z_1^2+ b_{1,1}z_1z_2 + b_{1,0}z_1 + b_{0,1} z_2 \right),\\
\label{eq:11fact2}
&\left(a_{2,0}z_1^2 + a_{1,0}z_1 + a_{0,0}\right)
\left(b_{2,1}z_1^2 z_2 + b_{1,2}z_1 z_2^2 + b_{1,1}z_1z_2 + b_{1,0}z_1 + b_{0,1} z_2\right),\\
\label{eq:11fact3}
&\left(a_{2,0}z_1^2 + a_{1,1}z_1z_2 +  a_{1,0}z_1 + a_{0,0}\right)
\left(b_{2,1}z_1^2 z_2 + b_{1,1}z_1z_2 + b_{1,0}z_1 + b_{0,1} z_2\right),\\
\label{eq:11fact4}
&\left(a_{2,0}z_1^2 + a_{1,1}z_1z_2 +  a_{1,0}z_1 +  a_{0,1}z_2\right)
\left(b_{2,1}z_1^2 z_2 + b_{1,1}z_1z_2 + b_{1,0}z_1 + b_{0,0}\right),\\
\label{eq:11fact5}
& \left(a_{1,0}z_1 +  a_{0,0}\right)
\left(b_{3,1}z_1^3 z_2 +b_{2,2}z_1^2 z_2^2 +  b_{2,1}z_1^2 z_2 + b_{1,2}z_1z_2^2+b_{2,0}z_1^2  +b_{1,1}z_1z_2  + b_{1,0}z_1 + b_{0,1}z_2\right),\\
\label{eq:11fact6}
& \left(a_{1,0}z_1 +  a_{0,1}z_2\right)
\left(b_{3,1}z_1^3 z_2 +b_{2,1}z_1^2 z_2 +b_{2,0}z_1^2  +b_{1,1}z_1z_2  + b_{1,0}z_1 + b_{0,0}\right).
\end{align}
Nonvanishing elements of $\text{Coeff}_E$ for $a_0\in \mathbb R$ and $S_\Gamma \in (\mathbb C^*)^6$ are, for these six cases respectively, $a_1 c_1 d_1 \bar d_2$, $\bar b_0^2 d_1 \bar d_2 \bar d_3^2$, $a_1 b_0^2 c_1 \bar d_1 d_2 d_3^2$, $\bar b_0^2 d_1^2 \bar d_2^2 \bar d_3^2$, $b_0^2 d_1^2 d_2^2 d_3^2$, and $\bar a_1 \bar b_0 \bar c_1 \bar d_3$.


\smallskip
For Cases B4 and B5, using the isomorphic graphs shown in Fig.~\ref{fig:GIso}, 
we obtain {\color{black}
$S_{\Gamma}=\left\{a_1,b_0,d_3, c_1\right\}$ and $S_{\Gamma}=\left\{a_1,b_0,d_1,d_3, c_1\right\}$}, respectively. 
By setting  {\color{black}$d_1=d_2=0$ or $d_1=0$} in \eqref{eq:Pc_11}, we find that these two cases possess a common $C_Z$ and a common subset $\tilde C_N$ of $C_N$,
\begin{equation}\label{eq:czcn10}
\begin{aligned}
&C_Z = \left\{ (5,0), (2,3), (1,4), (0,5), (4,0), (1,3), (0,4),(3,0), (1,2), (0,3), (0,2), (0,0)\right\}  \cup \left\{ (m,n)\in U: m+n>5  \right\} \\ 
&\tilde C_N = \left\{ (4,1), (3,2), (1,0), (0,1) \right\}.
\end{aligned}
\end{equation}
The constraints in \eqref{eq:czcn10} do not allow any reduced index pairs, and so these cases are done.

\smallskip
For Case C1, $S_{\Gamma}=\left\{a_1,a_2,b_0\right\}$ and
\begin{equation}
P_c = -\bar a_2 z_1 E - \bar a_1 z_2 E - a_1 z_1^2 z_2 E - a_2 z_1 z_2^2 E + z_1 z_2 (-b_0 \bar b_0  {\color{black}-a_0E}+ E^2).
 \end{equation}
%
For Case C2, $S_{\Gamma}=\left\{a_1,a_2,b_0,d_1\right\}$ and
\begin{equation}
P_c =  -\bar a_2 z_1 E - a_2 z_1 z_2^2 E + z_1^2 z_2 (-b_0 d_1 - a_1 E) + z_2 (-\bar b_0 \bar d_1 - \bar a_1 E) + z_1 z_2 (-b_0 \bar b_0 - d_1 \bar d_1  {\color{black}-a_0E} + E^2).
 \end{equation}
%
For Case B1, $S_{\Gamma}=\left\{a_1,b_0,d_2\right\}$ and
\begin{equation}
P_c = -\bar b_0 \bar d_2 z_1 - b_0 d_2 z_1 z_2^2 - \bar a_1 z_2 E - a_1 z_1^2 z_2 E + z_1 z_2 (-b_0 \bar b_0 - d_2 \bar d_2  {\color{black}-a_0E}+ E^2).
 \end{equation}
In each of these cases $P_c$ is of degree $3$, and 
\begin{equation}\label{eq:czcn4}
\begin{aligned}
&C_Z = \left\{ (3,0), (0,3), (2,0), (0,2), (0,0)\right\}  \cup \left\{ (m,n)\in U: m+n>3  \right\}, \\ 
&\tilde C_N = \left\{ (2,1), (1,2)\right\} \subset C_N.
\end{aligned}
\end{equation}
The constraints in \eqref{eq:czcn4} do not allow any reduced index pairs, so we have taken care of these three cases.

\subsection{Proof of Theorem~\ref{thm:irreducible}(b): Reducibility for the tetrakis graph}\label{sec:irreducible} 

We examine the factorizations in Case C8, which is the tetrakis graph.  We saw in the previous section that this is the only case in which the Fermi surface can be reducible.  Recall 
that reducibility can result from three possible factorizations (\ref{eq:C8fact1}, \ref{eq:C8fact2}, \ref{eq:C8fact3}).
To determine the graph coefficients that realize these factorizations, we seek common roots $\left(a_0,a_1,a_2,b_0,d_1,d_2,d_3\right)$ of all the elements of $\text{Coeff}_E$ for each of these three pairs of factors.  For pairs~\eqref{eq:C8fact2} and~\eqref{eq:C8fact3}, we find all roots with $a_0$ real and $\left(a_1,a_2,b_0,d_1,d_2,d_3\right)$ complex, and for pair~\eqref{eq:C8fact1}, we find all real roots $\left(a_0,a_1,a_2,b_0,d_1,d_2,d_3\right)$.  Each of these $7$-tuples of coefficients $\left(a_0,a_1,a_2,b_0,d_1,d_2,d_3\right)$ corresponds to a discrete graph whose Fermi surface is reducible.
The difficulty in finding all complex solutions to pair~\eqref{eq:C8fact1} is discussed at the end of this section.

\subsubsection{Complex roots of $\text{Coeff}_E$ for factorization~\eqref{eq:C8fact2}}\label{sec:C8fact2}

A subset of $\text{Coeff}_E$ for the factorization~\eqref{eq:C8fact2} in Case C8 is
\begin{equation}\label{eq:subcoeff2}
\begin{aligned}
\mathcal J = \big\{ -\bar a_1 \bar a_2 - \bar b_0 \bar d_3,\,
 \bar a_1 a_2 + \bar d_1 d_2,\,
 -\bar a_2 \bar b_0 \bar d_1 - \bar a_1 \bar b_0 \bar d_2  {\color{black}+ a_0} \bar b_0 \bar d_3 - \bar a_1 d_1 \bar d_3 - \bar a_2 d_2 \bar d_3,& \\
 -\bar b_0^2 \bar d_1 \bar d_2 + b_0 \bar b_0^2 \bar d_3 - d_1 d_2 \bar d_3^2 + \bar b_0 d_3 \bar d_3^2,\,
 -b_0^2 d_1 d_2 + b_0^2 \bar b_0 d_3 - \bar d_1 \bar d_2 d_3^2 + b_0 d_3^2 \bar d_3&\big\}.
 \end{aligned}
\end{equation}
A common root of $\text{Coeff}_E$ has to be a common root of~$\mathcal J$. We will find all roots of \eqref{eq:subcoeff2} and check that they are in fact roots of $\text{Coeff}_E$. From the first two entries in $\mathcal J$, $a_1 = -b_0 d_3/a_2$ and $d_1 = \bar a_2 b_0 d_3/(a_2 \bar d_2)$. {\color{black} Note that the condition for $a_1$ is not achievable if all non-loop edge weights are positive.}  By substituting these relations and their complex conjugates into \eqref{eq:subcoeff2}, the five elements of $\mathcal J$ become (in the same order)
 \begin{equation}\label{eq:subcoeff2a}
 \begin{aligned}
\mathcal J = \big\{  0, 0, \bar d_3 ( {\color{black}a_0} \bar b_0 - a_2 \bar b_0^2/d_2 - \bar a_2 d_2 + \bar b_0^2 \bar d_2/\bar a_2 + b_0 \bar b_0 d_3 \bar d_3/(a_2 \bar d_2)),& \\
-(-\bar a_2 b_0 d_2 + a_2 \bar b_0 \bar d_2) \bar d_3 (a_2 \bar b_0^2 \bar d_2 -  \bar a_2 d_2 d_3 \bar d_3)/(a_2 \bar a_2 d_2 \bar d_2),& \\
     -(-\bar a_2 b_0 d_2 + a_2 \bar b_0 \bar d_2) d_3 (-\bar a_2 b_0^2 d_2 + a_2 \bar d_2 d_3 \bar d_3)/(
 a_2 \bar a_2 d_2 \bar d_2)&\big\}.
 \end{aligned}
\end{equation}
Vanishing of the third entry in \eqref{eq:subcoeff2a} requires that $a_0 =  {\color{black}a_2 \bar b_0/d_2 + \bar a_2 d_2/ \bar b_0 - \bar b_0 \bar d_2/\bar a_2 - b_0  d_3 \bar d_3/(a_2 \bar d_2)}$ and that this expression be real (since $a_0$ must be real).  By making this substitution for $a_0$ in~\eqref{eq:subcoeff2a}, $\mathcal J$ is updated to
\begin{equation}\label{eq:subcoeff2b}
\begin{aligned}
\mathcal J = \big\{0, 0, 0, -(-\bar a_2 b_0 d_2 + a_2 \bar b_0 \bar d_2) \bar d_3 (a_2 \bar b_0^2 \bar d_2 - \bar a_2 d_2 d_3 \bar d_3)/(a_2 \bar a_2 d_2 \bar d_2),&\\
     -(-\bar a_2 b_0 d_2 + a_2 \bar b_0 \bar d_2) d_3 (-\bar a_2 b_0^2 d_2 + a_2 \bar d_2 d_3 \bar d_3)/(
 a_2 \bar a_2 d_2 \bar d_2)&\big\}.
\end{aligned}
\end{equation}  
For \eqref{eq:subcoeff2b} to vanish, either $\bar a_2 = a_2 \bar b_0 \bar d_2/(b_0 d_2)$ or $\bar a_2 = a_2 \bar d_2 d_3 \bar d_3/(b_0^2 d_2)$.  In the former case, \eqref{eq:subcoeff2b} identically vanishes and $ a_0 =  {\color{black}a_2 \bar b_0/d_2 + \bar a_2 d_2/ \bar b_0 - \bar b_0 \bar d_2/\bar a_2 - b_0  d_3 \bar d_3/(a_2 \bar d_2)} \in \mathbb R$.  Note that $\bar a_2 = a_2 \bar b_0 \bar d_2/(b_0 d_2)$ is a constraint on $a_2$ and $\bar a_2$ but does not eliminate $a_2$ and $\bar a_2$. We obtain that the substitutions
\begin{equation}\label{parameterization1}
\begin{aligned}
  a_0 &=  {\color{black}a_2 \bar b_0/d_2 + \bar a_2 d_2/ \bar b_0 -   \bar b_0 \bar d_2/\bar a_2 - b_0  d_3 \bar d_3/(a_2 \bar d_2)},\\
  a_1 &= -b_0 d_3/a_2,\\
  \bar a_2 &= a_2 \bar b_0 \bar d_2/(b_0 d_2),\\
  d_1 &= \bar a_2 b_0 d_3/(a_2 \bar d_2)
\end{aligned}
\end{equation}
cause \eqref{eq:subcoeff2} to vanish.
In the latter case, by substituting $\bar a_2 = a_2 \bar d_2 d_3 \bar d_3/(b_0^2 d_2)$ into \eqref{eq:subcoeff2b}, $\mathcal J$ is updated into
 \begin{equation}\label{eq:subcoeff2c}
\mathcal J = \big\{0, 0, 0, 0,
 -(b_0 \bar b_0 - d_3 \bar d_3)^2 (b_0 \bar b_0 + d_3 \bar d_3)/(\bar b_0 \bar d_3)\big\}.
\end{equation}  
Thus $b_0 \bar b_0 - d_3 \bar d_3=0$. In fact, the relations $\bar a_2 = a_2 \bar d_2 d_3 \bar d_3/(b_0^2 d_2)$ and $b_0 \bar b_0 - d_3 \bar d_3=0$ combine to give $\bar a_2 = a_2 \bar b_0 \bar d_2/(b_0 d_2)$, which turns the latter case ($\bar a_2 = a_2 \bar d_2 d_3 \bar d_3/(b_0^2 d_2)$) into a subcase of the former case ($\bar a_2 = a_2 \bar b_0 \bar d_2/(b_0 d_2)$).
We check in Mathematica{\textregistered} that $\text{Coeff}_E$ vanishes identically for \eqref{parameterization1}. Thus all complex solutions to $\text{Coeff}_E$ for $a_0\in \mathbb R$ and $S_\Gamma \in (\mathbb C^*)^6$ are represented by \eqref{parameterization1}.

We conclude that $P_c$ factors into the form~\eqref{eq:C8fact2} if and only if the relations~\eqref{parameterization1} are satisfied.  These relations provide a parameterization of all factorizations of $P_c$ by the four variables $\left\{a_2,b_0,d_2,d_3\right\}$.  The factorizations are
\begin{equation}
\begin{aligned}
|d_2|^2 P_c 
=-\big( \bar d_2 d_3 z_1^2 + d_2\bar d_3 + (|d_3|^2+|d_2|^2 + a_2 \bar d_2 E/b_0)z_1\big)
\big(b_0 d_2z_2^2 +\bar b_0\bar d_2 + (|d_2|^2+|b_0|^2 + b_0 d_2E/a_2)z_2\big).
 \end{aligned}
\end{equation}

\subsubsection{Complex roots of $\text{Coeff}_E$ for factorization~\eqref{eq:C8fact3}}\label{sec:C8fact3}

A subset of $\text{Coeff}_E$ is
\begin{equation}\label{eq:subcoeff3}
\begin{aligned}
\mathcal K = \big\{(\bar a_1 a_2 + \bar d_1 d_2) (\bar a_1 \bar a_2 + \bar b_0 \bar d_3), (-d_1 d_2 + \bar b_0 d_3) (\bar d_1 \bar d_2 - b_0 \bar d_3) (\bar b_0 \bar d_1 - d_2 \bar d_3)^2&\,, \\
-a_1 \bar a_1 d_1 \bar d_1 d_2 \bar d_2 + \bar a_1^2 b_0 d_1 \bar d_2 d_3 + a_1^2 \bar b_0 \bar d_1 d_2 \bar d_3 - 
 a_1 \bar a_1 b_0 \bar b_0 d_3 \bar d_3&\,,\\
 \bar a_2 b_0^3 d_1^2 d_2 - a_2 b_0 d_1^3 \bar d_1 \bar d_2 + a_1 b_0^2 \bar b_0 d_1 d_2 \bar d_2  
 {\color{black} - a_0} b_0^2 d_1^2 d_2 \bar d_2 - a_1 b_0 d_1^2 \bar d_1 d_2 \bar d_2 + a_2 b_0 d_1^2 d_2 \bar d_2^2
 - a_1 b_0 d_1 d_2^2 \bar d_2^2&\\ 
  - \bar a_2 b_0^3 \bar b_0 d_1 d_3  - a_1 b_0^2 \bar b_0^2 \bar d_2 d_3 
  {\color{black}+ a_0} b_0^2 \bar b_0 d_1 \bar d_2 d_3  {\color{black} + a_0} b_0 d_1^2 \bar d_1 \bar d_2 d_3 - 
 a_1 d_1^2 \bar d_1^2 \bar d_2 d_3 - \bar a_2 b_0^2 d_1 d_2 \bar d_2 d_3  &\\ 
 - a_2 b_0 \bar b_0 d_1 \bar d_2^2 d_3 + a_2 d_1^2 \bar d_1 \bar d_2^2 d_3+ 
 a_1 b_0 \bar b_0 d_2 \bar d_2^2 d_3  {\color{black}+ a_0} b_0 d_1 d_2 \bar d_2^2 d_3 - 
 a_1 d_1 \bar d_1 d_2 \bar d_2^2 d_3 - a_2 d_1 d_2 \bar d_2^3 d_3  &\\ 
 + \bar a_2 b_0^2 \bar b_0 \bar d_2 d_3^2 - 
 \bar a_2 b_0 d_1 \bar d_1 \bar d_2 d_3^2  {\color{black}- a_0} b_0 \bar b_0 \bar d_2^2 d_3^2 + 
 2 a_1 \bar b_0 \bar d_1 \bar d_2^2 d_3^2  {\color{black}- a_0} d_1 \bar d_1 \bar d_2^2 d_3^2 + a_2 \bar b_0 \bar d_2^3 d_3^2
 + \bar a_2 \bar d_1 \bar d_2^2 d_3^3 &\\ 
  + a_2 b_0^2 d_1^3 \bar d_3 + 2 a_1 b_0^2 d_1^2 d_2 \bar d_3 - 
 a_1 b_0^2 \bar b_0 d_1 d_3 \bar d_3  {\color{black}- a_0} b_0^2 d_1^2 d_3 \bar d_3 + a_1 b_0 d_1^2 \bar d_1 d_3 \bar d_3 -
  a_2 b_0 d_1^2 \bar d_2 d_3 \bar d_3  &\\ 
  + \bar a_2 b_0^2 d_1 d_3^2 \bar d_3 - 
 a_1 b_0 \bar b_0 \bar d_2 d_3^2 \bar d_3 
  {\color{black}+ a_0} b_0 d_1 \bar d_2 d_3^2 \bar d_3 + 
 a_1 d_1 \bar d_1 \bar d_2 d_3^2 \bar d_3 - \bar a_2 b_0 \bar d_2 d_3^3 \bar d_3 - a_1 b_0 d_1 d_3^2 \bar d_3^2&\,,\\
a_1 \bar a_1 a_2 \bar a_2 b_0 \bar b_0  {\color{black}+ a_0} \bar a_1 a_2 \bar a_2 b_0 d_1 
 {\color{black}+ a_0} a_1 a_2 \bar a_2 \bar b_0 \bar d_1 + 
 a_1 \bar a_1 a_2 \bar a_2 d_1 \bar d_1 - a_2 \bar a_2 b_0 \bar b_0 d_1 \bar d_1
   {\color{black}+ a_0} a_1 \bar a_1 \bar a_2 b_0 d_2 &\\
  - \bar a_1 \bar a_2 b_0^2 d_1 d_2  {\color{black}- a_0^2} a_1 \bar a_2 \bar d_1 d_2 + a_1 a_2 \bar a_2^2 \bar d_1 d_2 + 
 a_1 \bar a_2 b_0 \bar b_0 \bar d_1 d_2  {\color{black}+ 2 a_0} \bar a_2 b_0 d_1 \bar d_1 d_2 
 + 2 a_1 \bar a_2 d_1 \bar d_1^2 d_2 &\\
 {\color{black} + a_0} a_1 \bar a_1 a_2 \bar b_0 \bar d_2  {\color{black}- a_0^2} \bar a_1 a_2 d_1 \bar d_2 + \bar a_1 a_2^2 \bar a_2 d_1 \bar d_2 + 
 \bar a_1 a_2 b_0 \bar b_0 d_1 \bar d_2 - a_1 a_2 \bar b_0^2 \bar d_1 \bar d_2   {\color{black}+ 
 2 a_0} a_2 \bar b_0 d_1 \bar d_1 \bar d_2 &\\
 + 2 \bar a_1 a_2 d_1^2 \bar d_1 \bar d_2 + 
 a_1 \bar a_1 a_2 \bar a_2 d_2 \bar d_2 - a_1 \bar a_1 b_0 \bar b_0 d_2 \bar d_2  
 {\color{black} + 2 a_0} \bar a_1 b_0 d_1 d_2 \bar d_2  {\color{black}+ 2 a_0} a_1 \bar b_0 \bar d_1 d_2 \bar d_2 
  {\color{black}- 3 a_0^2} d_1 \bar d_1 d_2 \bar d_2 &\\
 + 2 a_2 \bar a_2 d_1 \bar d_1 d_2 \bar d_2  + 
 b_0 \bar b_0 d_1 \bar d_1 d_2 \bar d_2 + 3 d_1^2 \bar d_1^2 d_2 \bar d_2 + 2 a_1 \bar a_2 \bar d_1 d_2^2 \bar d_2 + 
 2 \bar a_1 a_2 d_1 d_2 \bar d_2^2 + 3 d_1 \bar d_1 d_2^2 \bar d_2^2 &\\
 - \bar a_1 a_2 \bar a_2^2 b_0 d_3 
  {\color{black}+  a_0} a_1 \bar a_1 \bar a_2 \bar d_1 d_3 - \bar a_1 \bar a_2 b_0 d_1 \bar d_1 d_3 - a_1 \bar a_2 \bar b_0 \bar d_1^2 d_3 - \bar a_2^2 b_0 \bar d_1 d_2 d_3  {\color{black}+ a_0} \bar a_1 a_2 \bar a_2 \bar d_2 d_3 &\\
 - a_1 \bar a_1 \bar b_0 \bar d_1 \bar d_2 d_3 - a_2 \bar a_2 \bar b_0 \bar d_1 \bar d_2 d_3 - 2 a_0 \bar a_1 d_1 \bar d_1 \bar d_2 d_3 - \bar b_0 d_1 \bar d_1^2 \bar d_2 d_3 - \bar a_1 \bar a_2 b_0 d_2 \bar d_2 d_3  {\color{black}+ 2 a_0} \bar a_2 \bar d_1 d_2 \bar d_2d_3 &\\
 -\bar a_1 a_2 \bar b_0 \bar d_2^2 d_3 - \bar b_0 \bar d_1 d_2 \bar d_2^2 d_3 - \bar a_1 \bar a_2 \bar d_1 \bar d_2 d_3^2 - a_1 a_2^2 \bar a_2 \bar b_0 \bar d_3  {\color{black}+ a_0} a_1 \bar a_1 a_2 d_1 \bar d_3 - \bar a_1 a_2 b_0 d_1^2 \bar d_3&\\
  -  a_1 a_2 \bar b_0 d_1 \bar d_1 \bar d_3  {\color{black}+ a_0} a_1 a_2 \bar a_2 d_2 \bar d_3 - a_1 \bar a_1 b_0 d_1 d_2 \bar d_3 - 
 a_2 \bar a_2 b_0 d_1 d_2 \bar d_3  {\color{black}+ 2 a_0} a_1 d_1 \bar d_1 d_2 \bar d_3 - b_0 d_1^2 \bar d_1 d_2 \bar d_3 &\\
 - a_1 \bar a_2 b_0 d_2^2 \bar d_3 - a_2^2 \bar b_0 d_1 \bar d_2 \bar d_3 - a_1 a_2 \bar b_0 d_2 \bar d_2 \bar d_3 
  {\color{black}  + 2 a_0} a_2 d_1 d_2 \bar d_2 \bar d_3 - b_0 d_1 d_2^2 \bar d_2 \bar d_3 + a_1 \bar a_1 a_2 \bar a_2 d_3 \bar d_3  &\\
 - a_1 \bar a_1 d_1 \bar d_1 d_3 \bar d_3
 + a_1 \bar a_2 \bar d_1 d_2 d_3 \bar d_3 + \bar a_1 a_2 d_1 \bar d_2 d_3 \bar d_3 - 
 a_2 \bar a_2 d_2 \bar d_2 d_3 \bar d_3 + d_1 \bar d_1 d_2 \bar d_2 d_3 \bar d_3 - a_1 a_2 d_1 d_2 \bar d_3^2\, &\big\}.
\end{aligned}
\end{equation}
A root of $\text{Coeff}_E$ has to be a root of $\mathcal K$. From the first entry in $\mathcal K$, $a_1 = -d_1 \bar d_2/\bar a_2$ or $a_1 = -b_0 d_3/a_2$. From the second entry in $\mathcal K$, $d_1 =\bar b_0 d_3/d_2$ or $d_1=\bar d_2 d_3/b_0$.  Next, following the approach in Sec.~\ref{sec:C8fact2}, it turns out that all roots of $\mathcal K$ take four forms; one of them is equivalent to \eqref{parameterization1}, and the others are subsets of the set defined by \eqref{parameterization1}.  A factorization of the form \eqref{eq:C8fact2} implies a factorization of the form \eqref{eq:C8fact3}.  Thus all roots of $\text{Coeff}_E$ for pair \eqref{eq:C8fact3}, and therefore all
factorizations of the form~\eqref{eq:C8fact3}, are given exactly by the parameterization~\eqref{parameterization1}.
This implies that each factorization of the form~\eqref{eq:C8fact3} is in fact a factorization of the form~\eqref{eq:C8fact2}.

\subsubsection{Real roots of $\text{Coeff}_E$ for factorization~\eqref{eq:C8fact1}}\label{sec:C8fact1}

Observe that, if $P_c$ factors into the form~\eqref{eq:C8fact2} for some graph operator, then it also factors into the form~\eqref{eq:C8fact1}.
We now seek all graphs for which $P_c$ factors into the form~\eqref{eq:C8fact1}.  As stated in Theorem~\ref{thm:irreducible}, we are able to parameterize all such graphs with real coefficients.

Let us briefly summarize the methods of the previous two subsections.
For pairs~\eqref{eq:C8fact2} and~\eqref{eq:C8fact3} of Case C8, the method for obtaining all roots of $\text{Coeff}_E$ for $a_0\in \mathbb R$ and $S_\Gamma \in (\mathbb C^*)^6$ goes as follows.  $\text{Coeff}_E$ is generically a set of polynomials in variables $\left\{x_i: 1\leq i\leq N\right\}$.  Suppose $\text{Coeff}_E$ contains an element that factors into $M$ factors and, for each $k:1\leq k\leq M$, when the $k$th factor is zero, one variable, say $x_{i_k}$, can be solved in terms of other variables (complex conjugates are treated as independent variables).  In this situation, finding roots of $\text{Coeff}_E$ is equivalent to considering $M$ cases where each factor is zero. In the $k$th case, we substitute for $x_{i_k}$ its expression in terms of other variables into all the elements of $\text{Coeff}_E$.  When this expression does not involve $\bar x_{i_k}$, we also substitute $\bar x_{i_k}$ by the conjugate of the expression for $x_{i_k}$. After this substitution, $\text{Coeff}_E$ is updated and the number of variables decreases. If this process can be repeated until all polynomials in $\text{Coeff}_E$ either identically vanish or one of them is identically nonzero (such as $a_1\bar a_1 + b_2\bar b_2$), then we have obtained all common roots of the elements of $\text{Coeff}_E$ and therefore all graphs that factor into the given form $P_c = P_aP_b$.


It turns out that, by this method and symmetry analysis in this section, all admissible graph operators with real coefficients and factorable dispersion function can be determined.  But the algebraic equations become prohibitively complicated when the graph coefficients are allowed to be complex, even with a combination of computer and by-hand computations.  Thus the only part of this study that remains incomplete is the {\color{black}determination} of all graph operators associated to the tetrakis graph that have non-real coefficients and for which the Fermi surface is reducible.

For  $E, a_0\in \mathbb R$ and $S_\Gamma \in (\mathbb R^*)^6$, we can obtain all factorizations of $P_c$ into pair~{\eqref{eq:C8fact1} by exploiting symmetry.  When all these variables are real, the symmetry
$D(z_1,z_2;E)=\bar D(z_1^{-1},z_2^{-1};E)$ becomes
\begin{equation}
D(z_1,z_2;E)=D(z_1^{-1},z_2^{-1};E).
\end{equation}
Suppose that $D(z_1,z_2;E)$ factors into two Laurent, non-monomial, polynomials
%
$
D(z_1,z_2;E)=A(z_1,z_2)B(z_1,z_2).
$
%
Then
\begin{equation}
D(z_1,z_2;E)=A(z_1,z_2)B(z_1,z_2)=A(z_1^{-1},z_2^{-1})B(z_1^{-1},z_2^{-1}).
\end{equation}
Multiplying by $z_1^{m_0}z_2^{n_0}$ as in Step 3 of the algorithm yields
\begin{equation}\label{realsymmetry}
P_c(z_1,z_2)=z_1^{m_0}z_2^{n_0}A(z_1,z_2)B(z_1,z_2)
=z_1^{m_0}z_2^{n_0}A(z_1^{-1},z_2^{-1})B(z_1^{-1},z_2^{-1}).
\end{equation}
%
%
%
If $P_c$ factors into pair~\eqref{eq:C8fact1}, then~(\ref{realsymmetry}) implies
\begin{equation}\label{twofactorizations}
P_c=F_1F_2=F_3F_4,
\end{equation}
in which
\begin{align}
F_1&= a_{1,1}z_1z_2+ a_{1,0}z_1+ a_{0,1}z_2+a_{0,0}, \\
F_2&= b_{1,1}z_1z_2+ b_{1,0}z_1+ b_{0,1}z_2+b_{0,0}, \\
F_3&= a_{1,1}+ a_{1,0}z_2+ a_{0,1}z_1+a_{0,0}z_1z_2, \\
F_4&=b_{1,1}+ b_{1,0}z_2+ b_{0,1}z_1+b_{0,0}z_1z_2.
\end{align}
If one of $F_1$, $F_2$, $F_3$ and $F_4$ factors, then it is reduced to pair~\eqref{eq:C8fact3}.  All the graphs in this case are parameterized by~(\ref{parameterization1}) with the values of $\left\{a_2,b_0,d_2,d_3\right\}$ being real.

Thus we need only to consider the situation in which $F_1$, $F_2$, $F_3$ and $F_4$ are all irreducible.  In this case, we have $F_3|F_1$ or else $F_3|F_2$, since $\mathbb C[x,y]$ is a unique-factorization domain. 

\begin{claim}
$F_3 \nmid F_1$
\end{claim}
\begin{proof}
The assumption that $F_3|F_1$ implies that $(a_{1,1}, a_{1,0}, a_{0,1},a_{0,0})=C( a_{0,0},a_{0,1},a_{1,0},a_{1,1})$.  This leads to $C^2=1$ and therefore $C=\pm1$. 
Thus $P_c$ factors into either
\begin{equation}\label{eq:8more11}
\begin{aligned}
P_a &= a_{0,0}z_1z_2+ a_{0,1}z_1+ a_{0,1}z_2+a_{0,0} \\
P_b &=b_{1,1}z_1z_2+ b_{1,0}z_1+ b_{0,1}z_2+b_{0,0}
\end{aligned}
\end{equation}
or else
\begin{equation}\label{eq:8more12}
\begin{aligned}
P_a &= -a_{0,0}z_1z_2 - a_{0,1}z_1+ a_{0,1}z_2+a_{0,0} \\
P_b &=b_{1,1}z_1z_2+ b_{1,0}z_1+ b_{0,1}z_2+b_{0,0}.
\end{aligned}
\end{equation}
Performing Steps 5-8 of the algorithm for the pair~\eqref{eq:8more11} reveals that $\text{Coeff}_E$ contains the element $-1$; and performing the algorithm for the pair~\eqref{eq:8more12} reveals that $\text{Coeff}_E$ contains the element $1$.  In either case, $P_c$ is not factorable, which is inconsistent with the assumption~(\ref{twofactorizations}).
\end{proof}

In the case in which $F_1$, $F_2$, $F_3$ and $F_4$ are all irreducible and $F_3|F_2$, we have $(a_{1,1}, a_{1,0}, a_{0,1},a_{0,0})= C (b_{0,0},  b_{0,1},b_{1,0},b_{1,1})$ for some complex $C$. That is,
$P_c$ factors into 
\begin{equation}\label{eq:8more13}
\begin{aligned}
P_a &= C(b_{0,0}z_1z_2+ b_{0,1}z_1+ b_{1,0}z_2+b_{1,1}), \\
P_b &=b_{1,1}z_1z_2+ b_{1,0}z_1+ b_{0,1}z_2+{\color{black}b_{0,0}}.
\end{aligned}
\end{equation}
Performing the algorithm for the pair~\eqref{eq:8more13}, we obtain that $\text{Coeff}_E$ consists of the following four polynomials in the variables $(a_0,a_1,a_2,b_0,d_1,d_2,d_3)$:
\begin{equation}\label{eq:coeff8more1}
\begin{aligned}
\text{Coeff}_E = \big\{(a_1 a_2 + d_1 d_2) (a_1 a_2 + b_0 d_3)&\,,\\
2 a_1  a_2^2 b_0 d_1 + 2 a_1^2 a_2 b_0 d_2  {\color{black}- a_0} a_1 a_2 d_1 d_2 + a_2 b_0 d_1^2 d_2 + 
  a_1 b_0 d_1 d_2^2  {\color{black}- a_0} a_1 a_2 b_0 d_3 + 2 a_1^2 a_2 d_1 d_3 + a_2 b_0^2 d_1 d_3\\
   + 2 a_1 a_2^2 d_2 d_3 + a_1 b_0^2 d_2 d_3  {\color{black}- 2 a_0} b_0 d_1 d_2 d_3 + a_1 d_1^2 d_2 d_3 +
   a_2 d_1 d_2^2 d_3 + a_1 b_0 d_1 d_3^2 + a_2 b_0 d_2 d_3^2&\,,\\
 a_2^2 b_0^2 d_1^2 + 3 a_1 a_2 b_0^2 d_1 d_2  {\color{black}- a_0} a_2 b_0 d_1^2 d_2 - 
  a_1 a_2 d_1^3 d_2 + a_1^2 b_0^2 d_2^2  {\color{black}- a_0} a_1 b_0 d_1 d_2^2 + a_1^2 d_1^2 d_2^2 +
   a_2^2 d_1^2 d_2^2 \\
   + b_0^2 d_1^2 d_2^2 - a_1 a_2 d_1 d_2^3 - a_1 a_2 b_0^3 d_3 
    {\color{black}- a_0} a_2 b_0^2 d_1 d_3 + 3 a_1 a_2 b_0 d_1^2 d_3  {\color{black} - a_0} a_1 b_0^2 d_2 d_3 
     {\color{black}+ a_0^2} b_0 d_1 d_2 d_3 - b_0^3 d_1 d_2 d_3\\
    {\color{black}- a_0} a_1 d_1^2 d_2 d_3 - b_0 d_1^3 d_2 d_3 + 3 a_1 a_2 b_0 d_2^2 d_3 
     {\color{black}- a_0} a_2 d_1 d_2^2 d_3 - 
  b_0 d_1 d_2^3 d_3 + a_1^2 b_0^2 d_3^2 + a_2^2 b_0^2 d_3^2  {\color{black}- a_0} a_1 b_0 d_1 d_3^2\\
   + a_1^2 d_1^2 d_3^2 + b_0^2 d_1^2 d_3^2  {\color{black}- a_0} a_2 b_0 d_2 d_3^2 + 
  3 a_1 a_2 d_1 d_2 d_3^2 + a_2^2 d_2^2 d_3^2 + b_0^2 d_2^2 d_3^2 + 
  d_1^2 d_2^2 d_3^2 - a_1 a_2 b_0 d_3^3 - b_0 d_1 d_2 d_3^3&\,,\\ 
  a_2 b_0^3 d_1^2 d_2 - a_2 b_0 d_1^4 d_2 + a_1 b_0^3 d_1 d_2^2 
   {\color{black} - a_0} b_0^2 d_1^2 d_2^2 + a_1 b_0 d_1^3 d_2^2 + a_2 b_0 d_1^2 d_2^3 - 
  a_1 b_0 d_1 d_2^4 - a_2 b_0^4 d_1 d_3 + a_2 b_0^2 d_1^3 d_3\\
   - a_1 b_0^4 d_2 d_3 
    {\color{black}+ a_0} b_0^3 d_1 d_2 d_3 - 2 a_1 b_0^2 d_1^2 d_2 d_3  {\color{black} + a_0} b_0 d_1^3 d_2 d_3 - 
  a_1 d_1^4 d_2 d_3 - 2 a_2 b_0^2 d_1 d_2^2 d_3 + a_2 d_1^3 d_2^2 d_3 + 
  a_1 b_0^2 d_2^3 d_3\\
    {\color{black}+ a_0} b_0 d_1 d_2^3 d_3 + a_1 d_1^2 d_2^3 d_3 - 
  a_2 d_1 d_2^4 d_3 + a_1 b_0^3 d_1 d_3^2  {\color{black}- a_0} b_0^2 d_1^2 d_3^2 + 
  a_1 b_0 d_1^3 d_3^2 + a_2 b_0^3 d_2 d_3^2 - 2 a_2 b_0 d_1^2 d_2 d_3^2 \\
   {\color{black}- a_0} b_0^2 d_2^2 d_3^2 - 2 a_1 b_0 d_1 d_2^2 d_3^2  {\color{black}- a_0} d_1^2 d_2^2 d_3^2 + 
  a_2 b_0 d_2^3 d_3^2 + a_2 b_0^2 d_1 d_3^3 + a_1 b_0^2 d_2 d_3^3
   {\color{black} + a_0} b_0 d_1 d_2 d_3^3 + a_1 d_1^2 d_2 d_3^3\\
   + a_2 d_1 d_2^2 d_3^3 - 
  a_1 b_0 d_1 d_3^4 - a_2 b_0 d_2 d_3^4
&\,\big\}.
\end{aligned}
\end{equation}
We will find all simultaneous real roots $(a_0,a_1,a_2,b_0,d_1,d_2,d_3)$ to the set~\eqref{eq:coeff8more1}. 
The vanishing of the first element in $\text{Coeff}_E$ implies that $a_1 = -d_1 d_2/a_2$ or $a_1 = -b_0 d_3/a_2$.
Updated with  $a_1 = -d_1 d_2/a_2$, the second element in $\text{Coeff}_E$ becomes 
\begin{equation}
a_2^{-1}(-d_1 d_2 + b_0 d_3) \big(  {\color{black} -a_0} a_2 d_1 d_2 - d_1 d_2 (b_0 d_2 + d_1 d_3) + a_2^2 (b_0 d_1 + d_2 d_3)\big).
\end{equation}
Thus $d_1 = b_0 d_3/d_2$ or $a_0 =  {\color{black}(-d_1 d_2 (b_0 d_2 + d_1 d_3) + a_2^2 (b_0 d_1 + d_2 d_3))/(a_2 d_1 d_2)}$.  With these two substitutions separately, the third element becomes
\begin{equation}\label{thirda}
d_3^2 \big(  {\color{black} -a_0} a_2 b_0 d_2 + a_2^2 (b_0^2 + d_2^2) - b_0^2 (d_2^2 + d_3^2)\big)^2/(a_2^2 d_2^2)
\end{equation}
or
\begin{equation}\label{thirdb}
(a_2^2 + d_1^2) (a_2^2 + d_2^2) (d_1 d_2 - b_0 d_3)^2/a_2^2,
\end{equation}
respectively.
By setting (\ref{thirda}) to zero and solving for $a_0$ or setting (\ref{thirdb}) to zero and solving for $d_1$, one obtains exactly the same expressions for $a_0$ and $d_1$.
Thus in the case $a_1 = -d_1 d_2/a_2$, we must have 
\begin{equation}\label{parameterization2}
\begin{aligned}
a_0 &=  {\color{black}(a_2^2 (b_0^2 + d_2^2) - b_0^2 (d_2^2 + d_3^2))/(a_2 b_0 d_2)}, \\
a_1 &= -b_0 d_3/a_2,\\
d_1 &= b_0 d_3/d_2.
\end{aligned}
\qquad \text{(parametrization for real coefficients)}
\end{equation} 
Similarly in the case that $a_1 = -b_0 d_3/a_2$, we obtain exactly these same expressions for $d_1$ and $a_0$ (for this, the fourth element of $\text{Coeff}_E$ is used). Thus, \eqref{parameterization2} holds in any case.  It can be checked that all elements of $\text{Coeff}_E$ vanish when \eqref{parameterization2} holds. Thus all solutions to $\text{Coeff}_E$ for $a_0\in \mathbb R$ and $S_\Gamma \in (\mathbb R^*)^6$ are given by~\eqref{parameterization2}.

Summarizing, we have proved in the case of real graph coefficients that if the polynomial $P_c$ admits a factorization of the form~(\ref{eq:C8fact1}) and the polynomials $F_1$, $F_2$, $F_3$ and $F_4$ are irreducible, then the graph coefficients must satisfy the relations (\ref{parameterization2}).  One finds that the set of real tuples $(a_0,a_1,a_2,b_0,d_1,d_2,d_3)$ that satisfy (\ref{parameterization2}) coincide exactly with the relations (\ref{parameterization1}) when all variables are real.  Those relations correspond to the graphs for which one (and therefore all) of the $F_i$ is reducible, and, as mentioned above, such graphs are subsumed by the factorization form~(\ref{eq:C8fact3}) and therefore also of the form~(\ref{eq:C8fact2}).  This implies that the assumption that all $F_i$ are irreducible is untenable.

This concludes the proof that all graphs with real coefficients and reducible Fermi surface must admit a factorization of the form~(\ref{eq:C8fact2}) and that all such factorizations are parameterized by four real variables according to~(\ref{parameterization2}).  Explicitly, the factorizations are
\begin{equation}
\begin{aligned}
d_2^2 P_c 
=-\big( d_2 d_3 z_1^2 + d_2 d_3 + (d_3^2+d_2^2 + a_2 d_2 E/b_0)z_1\big)
\big(b_0 d_2z_2^2 +b_0d_2 + (d_2^2+b_0^2 + b_0 d_2E/a_2)z_2\big),
 \end{aligned}
\end{equation}
and the corresponding factorizations of the dispersion function are
\begin{equation}
\begin{aligned}
D(z_1,z_2;E)
=-\left( d_3 (z_1 + z_1^{-1}) + \textstyle\frac{d_3^2}{d_2}+d_2 + \textstyle\frac{a_2}{b_0} E \right)
    \left(b_0 (z_2 + z_2^{-1}) + d_2+\textstyle\frac{b_0^2}{d_2} + \textstyle\frac{b_0}{a_2} E \right).
 \end{aligned}
\end{equation}

\smallskip
We end this section with a short discussion on complex graph coefficients, particularly, explaining why symmetry does not help with the search for all factorizations of $P_c$.
A factorization of the form~\eqref{eq:C8fact1} implies eight variables in the polynomials in the set $I$ in Step~5 of the algorithm.
When $a_0\in \mathbb R$ and $S_\Gamma \in (\mathbb C^*)^6$, this factorization, together with the symmetry $D(z_1,z_2;E)=\bar D(z_1^{-1},z_2^{-1};E)$ implies $P_c=F_1F_2=F_3F_4$ with
\begin{align}
F_1&=z_1^{m_1}z_2^{n_1}A(z_1,z_2) = a_{1,1}z_1z_2+ a_{1,0}z_1+ a_{0,1}z_2+a_{0,0}, \\
F_2&=z_1^{m_2}z_2^{n_2}B(z_1,z_2) = b_{1,1}z_1z_2+ b_{1,0}z_1+ b_{0,1}z_2+b_{0,0}, \\
F_3&=z_1^{m_3}z_2^{n_3}A(z_1^{-1},z_2^{-1}) = \bar a_{1,1}+ \bar a_{1,0}z_2+ \bar a_{0,1}z_1+\bar a_{0,0}z_1z_2, \\
F_4&=z_1^{m_4}z_2^{n_4}B(z_1^{-1},z_2^{-1})=\bar b_{1,1}+ \bar b_{1,0}z_2+ \bar b_{0,1}z_1+\bar b_{0,0}z_1z_2.
\end{align}
If $F_1$, $F_2$, $F_3$ and $F_4$ are all irreducible, then either $F_3|F_1$ and $F_4|F_2$ or $F_3|F_2$ and $F_4|F_1$. In the case $F_3|F_2$ and $F_4|F_1$, we obtain
\begin{equation}
C (\bar a_{1,1}, \bar a_{1,0}, \bar a_{0,1},\bar a_{0,0})=  (b_{0,0},  b_{0,1},b_{1,0},b_{1,1}),
\end{equation}
where $C$ is some complex number.
This means that $P_c$ factors into 
\begin{equation}\label{eq:8more}
\begin{aligned}
P_a &= a_{1,1}z_1z_2+ a_{1,0}z_1+ a_{0,1}z_2+a_{0,0} \\
P_b &=C(\bar a_{0,0}z_1z_2+ \bar a_{0,1}z_1+ \bar a_{1,0}z_2+\bar a_{1,1}).
\end{aligned}
\end{equation}
Since complex conjugates must be treated as independent variables in the process of variable elimination using Groebner-basis routines, the symmetry has not reduced the number of variables in the polynomials of the set $I$ in Step~5 of the algorithm.  
In fact, $GB_C(I)$ has not changed at all.
Determining all factorizations of the form~\eqref{eq:C8fact1} for the tetrakis graph with complex coefficients remains unsolved.

\section{Necessity of positivity and planarity: Reducible Fermi surfaces}

The examples in this section demonstrate that the positivity of the coefficients and the planarity conditions for irreducibility of the Fermi surface cannot be dispensed with.

\subsection{A non-planar graph with arbitrary coefficients}\label{subsec:nonplanar}

Non-planar graphs obtained by coupling two identical copies of a given periodic graph provide easy examples of reducible Fermi surfaces, as demonstrated in~\cite{Shipman2014}.
Fig.~\ref{fig:reducible2} depicts two identical copies of a (planar) square graph coupled by edges between corresponding vertices.  The resulting graph is not planar.
The matrices $A(n_1,n_2)$ for this example are
\begin{equation}
  A(0,0) =
  \mat{1.1}{0}{b}{\bar b}{0},
  \quad
  A(1,0) =
  \mat{1.1}{a}{0}{0}{a},
  \quad
  A(0,1) =
  \mat{1.1}{c}{0}{0}{c},
\end{equation}
plus $A(-1,0)=A(1,0)^*$ and $A(0,-1)=A(0,1)^*$ being determined by self-adjointness and the rest vanishing.  The Floquet transform for the operator $A-E$ is
\begin{equation}
  \hat A(z_1,z_2)-E \;=\;
  \mat{1.4}
  {az_1+\bar az_1^{-1}+cz_2+\bar cz_2^{-1}-E}
  {b}
  {\bar b}
  {az_1+\bar az_1^{-1}+cz_2+\bar cz_2^{-1}-E},
\end{equation}
and its determinant is
\begin{equation}
  \det ({\color{black}\hat A}(z_1,z_2)-E) \;=\;
  \left( az_1+\bar az_1^{-1}+cz_2+\bar cz_2^{-1} + |b| -E \right)
  \left( az_1+\bar az_1^{-1}+cz_2+\bar cz_2^{-1} - |b| -E \right),
\end{equation}
which demonstrates the reducibility of the Fermi surface for all energies~$E$.
This type of graph is a case of a more general construction of coupled graph operators with reducible Fermi surface~\cite{Shipman2014}.

\begin{figure}[ht]
\centerline{\scalebox{0.25}{\includegraphics{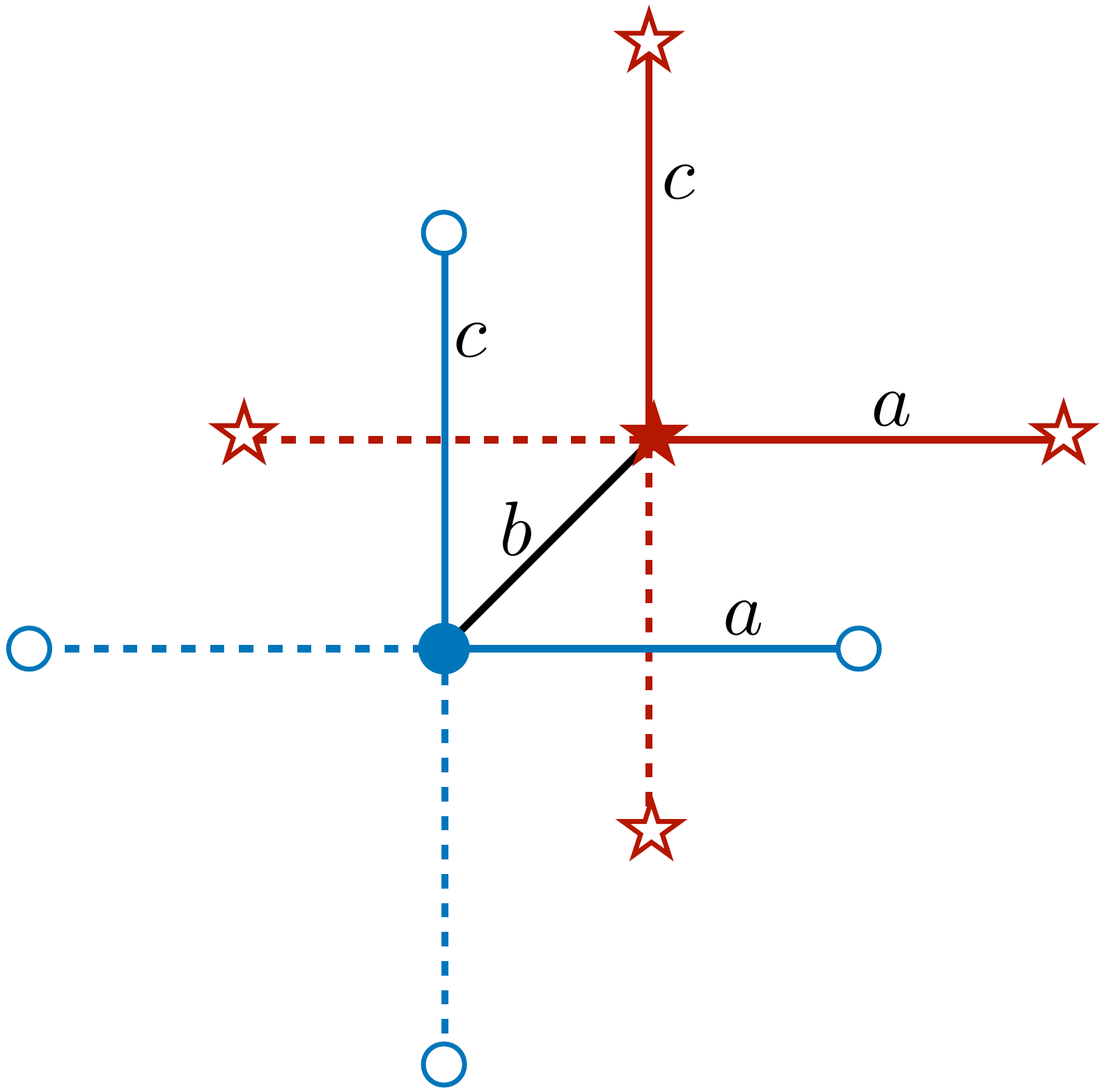}}}
\caption{\small A periodic non-planar graph with reducible Fermi surface.  The solid vertices and edges form a fundamental domain.}
\label{fig:reducible2}
\end{figure}

\subsection{Planar graphs with indefinite coefficients}\label{subsec:indefinite}

According to Theorem~\ref{thm:irreducible}(a), the Fermi surface of admissible graphs with positive coefficients is always irreducible.   
But for the tetrakis graph, factorizations of the form (\ref{d1d2}) described in Theorem~\ref{thm:irreducible}(b) and parameterized in section~\ref{sec:C8fact2} can be obtained with indefinite graph coefficients.  Fig.~\ref{fig:reducible1} shows some of these graphs with coefficients equal to $\pm1$.  The associated factorizations of the dispersion function are
\begin{align*}
  D_{(a)}(E,z_1,z_2) &\,=\, \left(z_1+z_1^{-1}+2-E\right)\left(z_2+z_2^{-1}+2+E\right), \\
  D_{(b)}(E,z_1,z_2) &\,=\, \left(z_1+z_1^{-1}+2-E\right)\left(z_2+z_2^{-1}-2-E\right), \\
  D_{(c)}(E,z_1,z_2) &\,=\, \left(z_1+z_1^{-1}-2+E\right)\left(z_2+z_2^{-1}-2-E\right).
\end{align*}
These three graph operators are unitarily equivalent by gauge transformation since they have identical circuit fluxes, by extending \cite[Lemma\,2.1]{LiebLoss1993a} to periodic graph operators.  Their Fermi surfaces are transformed by $z_1\mapsto z_1e^{i\phi_1}$ and $z_2\mapsto z_2e^{i\phi_1}$, where $\phi_j$ is the phase advance along the $j$-th period vector.  Each of these phases is $0$ or $\pi$ (mod $2\pi$), depending on the sign of the coefficient along the horizontal or vertical edge.  This implies that the spectra of these three operators are identical; in each case, the first factor of $D_{(j)}$ ($j\in\{a,b,c\}$) contributes the interval $E\in[0,4]$, and the second contributes $E\in[-4,0]$.  The dependence of the spectrum of periodic discrete magnetic operators on fluxes is studied in~\cite{KorotyaevSaburova2017}.

\begin{figure}[ht]
\centerline{\scalebox{0.34}{\includegraphics{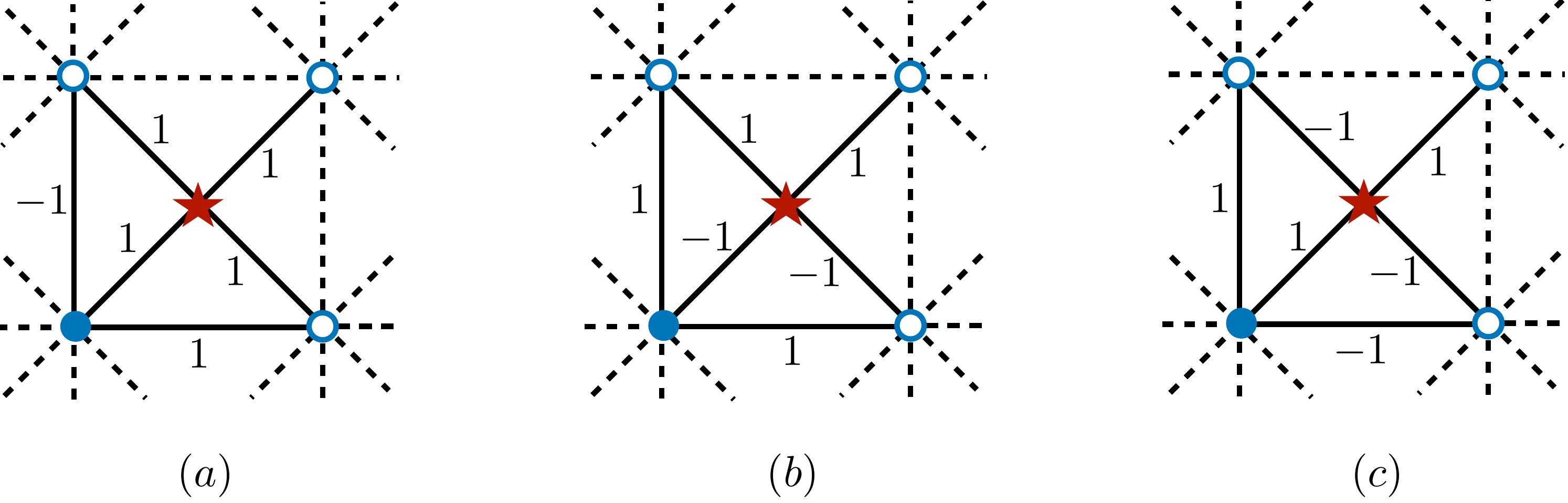}}}
\caption{\small Doubly periodic planar graphs with indefinite coefficients and reducible Fermi surface.  The solid vertices and edges form a fundamental domain, and $\pm1$ are coupling coefficients, as described in section~\ref{subsec:indefinite}.}
\label{fig:reducible1}
\end{figure}

\bigskip

\bigskip
\noindent
{\bfseries Acknowledgment.}  This work was supported by NSF research Grants DMS-1411393 and DMS-1814902.

\end{document}